\documentclass[11pt, a4paper]{article}
\usepackage[T1]{fontenc}
\usepackage{amsfonts}
\usepackage{amsmath}
\usepackage{amssymb}
\usepackage{amsthm}
\usepackage{bbm}
\usepackage{bm}
\usepackage{mathrsfs}
\usepackage{verbatim}
\usepackage{setspace}
\usepackage{enumitem}
\theoremstyle{plain}

\newtheorem{theorem}{Theorem}[section]
\newtheorem{proposition}[theorem]{Proposition}
\newtheorem{lemma}[theorem]{Lemma}
\newtheorem{corollary}[theorem]{Corollary}

\theoremstyle{definition}
\newtheorem{definition}[theorem]{Definition}
\newtheorem{remark}[theorem]{Remark}

\newtheorem{example}[theorem]{Example}

\newtheorem{assumption}[theorem]{Assumption}
\theoremstyle{remark}

\renewenvironment{thebibliography}[1]{%
\begin{oldthebibliography}{#1}%
\setlength{\baselineskip}{.9em}
\linespread{1}%
\small
\setlength{\parskip}{0ex}%
\setlength{\itemsep}{.2em}%
}%
{%
\end{oldthebibliography}%
}
\newcommand{\eps}{\varepsilon}

\newcommand{\Q}{\mathbb{Q}}
\newcommand{\R}{\mathbb{R}}

\newcommand{\E}{\mathbb{E}}
\newcommand{\F}{\mathbb{F}}
\newcommand{\Fci}{\mathbb{F}^\circ}
\newcommand{\cF}{\mathcal{F}}
\newcommand{\cFci}{\mathcal{F}^\circ}
\newcommand{\cN}{\mathcal{N}}

\newcommand{\cP}{\mathcal{P}}

\newcommand{\cA}{\mathcal{A}}
\newcommand{\cH}{\mathcal{H}}

\newcommand{\cE}{\mathcal{E}}
\newcommand{\cT}{\mathcal{T}}

\newcommand{\hF}{\mathbbm{F}}
\newcommand{\hcF}{\mathcal{F}}

\newcommand{\br}[1]{\langle #1 \rangle}

\DeclareMathOperator{\UC}{\textstyle{UC}}
\DeclareMathOperator{\esssup}{ess\, sup}
\DeclareMathOperator{\essinf}{ess\, inf}

\newcommand{\as}{\mbox{-a.s.}}
\newcommand{\qs}{\mbox{-q.s.}}

\newcommand{\tbE}{\widetilde{\E}}

\newcommand{\1}{\mathbf{1}}

\numberwithin{equation}{section}

\begin{document}

\title{Superhedging and Dynamic Risk Measures under Volatility Uncertainty\\
\date{First version: November 12, 2010. This version: June 2, 2012.}
\author{
  Marcel Nutz%
  \thanks{
  Department of Mathematics, Columbia University, \texttt{mnutz@math.columbia.edu}
  }
  \and
  H. Mete Soner%
  \thanks{
  Department of Mathematics, ETH Zurich, and Swiss Finance Institute, \texttt{mete.soner@math.ethz.ch}
  }
 }
}
\maketitle \vspace{-1em}

\begin{abstract}
We consider dynamic sublinear expectations (i.e., time-consistent coherent risk measures) whose scenario sets consist of singular measures corresponding to a general form of volatility uncertainty. We derive a c\`adl\`ag nonlinear martingale which is also the value process of a superhedging problem. The superhedging strategy is obtained from a representation similar to the optional decomposition. Furthermore, we prove an optional sampling theorem for the nonlinear martingale and characterize it as the solution of a second order backward SDE. The uniqueness of dynamic extensions of static sublinear expectations is also studied.
\end{abstract}

\vspace{.5em}

{\small


}\\

\noindent \emph{Acknowledgements}
Research supported by the
European Research Council Grant 228053-FiRM, the Swiss National Science Foundation
Grant PDFM2-120424/1 and the ETH Foundation. The authors thank two anonymous referees for helpful comments.

\section{Introduction}

Coherent risk measures were introduced in \cite{ArtznerDelbaenEberHeath.99} as a way to quantify the risk associated with a financial position. Since then, coherent risk measures and sublinear expectations (which are the same up to the sign convention) have been studied by numerous authors; see \cite{FollmerSchied.04, Peng.10icm, Peng.10} for extensive references. Most of these works consider the case where scenarios are probability measures absolutely continuous with respect to a given reference probability (important early exceptions are~\cite{FollmerSchied.02,Peng.05}). The present paper studies dynamic sublinear expectations and superhedging under volatility uncertainty, which is naturally related to singular measures. The concept of volatility uncertainty was introduced in financial mathematics by \cite{AvellanedaLevyParas.95, DenisMartini.06, Lyons.95} and has recently received considerable attention due to its relation to $G$-expectations~\cite{Peng.07, Peng.08} and second order backward stochastic differential equations~\cite{CheriditoSonerTouziVictoir.07, SonerTouziZhang.2010bsde}, called 2BSDEs for brevity.
Any (static) sublinear expectation $\cE^\circ_0$, defined on the set of bounded measurable functions on a measurable space $(\Omega,\cF)$, has a convex-dual representation
\begin{equation}\label{eq:staticRobustRep}
  \cE^\circ_0(X)=\sup_{P\in \cP} E^P[X]
\end{equation}
for a certain set $\cP$ of measures which are $\sigma$-additive as soon as $\cE^\circ_0$ satisfies certain continuity properties (cf.\ \cite[Section~4]{FollmerSchied.04}). The elements of $\cP$ can be seen as possible scenarios in the presence of uncertainty and hence~\eqref{eq:staticRobustRep} corresponds to the worst-case expectation. In this paper, we take $\Omega$ to be the canonical space of continuous paths and $\cP$ to be a set of martingale laws for the canonical process, corresponding to different scenarios of volatilities. For this case, $\cP$ is typically not dominated by a finite measure and~\eqref{eq:staticRobustRep} was studied
in~\cite{BionNadalKervarec.10, DenisHuPeng.2010, DenisMartini.06} by capacity-theoretic methods.
We remark that from the pricing point of view, the restriction to the martingale case entails no loss of generality in an arbitrage-free setting. An example with arbitrage was studied in~\cite{FernholzKaratzas.11}.

While any set of martingale laws gives rise to a static sublinear expectation via~\eqref{eq:staticRobustRep}, we are interested in \emph{dynamic} sublinear expectations; i.e., conditional versions of~\eqref{eq:staticRobustRep} satisfying a time-consistency property. If $\cP$ is dominated by a probability $P_*$, a natural extension of~\eqref{eq:staticRobustRep}  is given by
\[
  \cE^{\circ,P_*}_t(X)=\mathop{\esssup^{P_*}}_{P'\in \cP(\cFci_t,P_*)} E^{P'}[X|\cFci_t]\quad P_*\as,
\]
where $\cP(\cFci_t,P_*)=\{P'\in\cP:\, P'=P_*\mbox{ on }\cFci_t\}$ and $\Fci=\{\cFci_t\}$ is the filtration generated by the canonical process. Such dynamic expectations are well-studied; in particular, time consistency of $\cE^{\circ,P_*}$ can be characterized by a stability property of $\cP$ (see~\cite{Delbaen.06}).
In the non-dominated case, we can similarly consider the family of random variables $\{\cE^{\circ,P}_t(X),\, P\in\cP\}$.
Since a reference measure is lacking, it is not straightforward to construct a single random variable $\cE^{\circ}_t(X)$ such that
\begin{equation}\label{eq:aggreg}
  \cE^{\circ}_t(X)=\cE^{\circ,P}_t(X):=\mathop{\esssup^{P}}_{P'\in \cP(\cFci_t,P)} E^{P'}[X|\cFci_t]\quad P\as \quad\mbox{for all } P\in\cP.
\end{equation}
This problem of aggregation has been solved in several examples. In particular, the
$G$-expectations and random $G$-expectations~\cite{Nutz.10Gexp} (recalled in Section~\ref{se:preliminaries}) correspond to special cases of~\eqref{eq:aggreg}. The construction of $G$-expectations is based on a PDE, which directly yields random variables defined for all $\omega\in\Omega$. The random $G$-expectations are defined pathwise using regular conditional probability distributions. A general study of aggregation problems is presented in~\cite{SonerTouziZhang.2010aggreg}; see also~\cite{BionNadalKervarec.10b}. However, the study of aggregation is not an object of the present paper. In view of the diverse approaches, we shall proceed axiomatically and start with a given aggregated family $\{\cE^{\circ}_t(X),\,t\in[0,T]\}$. Having in mind the example of (random) $G$-expectations, this family is assumed to be given in the raw filtration $\F^\circ$ and without any regularity in the time variable.

The main goal of the present paper is to provide basic technology for the study of dynamic sublinear expectations under volatility uncertainty as stochastic processes. Given the family $\{\cE^{\circ}_t(X),\,t\in[0,T]\}$, we construct a corresponding c\`adl\`ag process $\cE(X)$, called the $\cE$-martingale associated with~$X$, in a suitably enlarged filtration $\F$ (Proposition~\ref{pr:extension}). We use this process to define the sublinear expectation at stopping times and prove an optional sampling theorem for $\cE$-martingales (Theorem~\ref{th:DPPstop}). Furthermore, we obtain a decomposition of $\cE(X)$ into an integral of the canonical process and an increasing process (Proposition~\ref{pr:martingaleDecomp}), similarly as in the classical optional decomposition \cite{ElKarouiQuenez.95}. In particular, the $\cE$-martingale yields the dynamic superhedging price of the financial claim $X$ and the integrand $Z^X$ yields the superhedging strategy. We also provide a connection between $\cE$-martingales and 2BSDEs by characterizing $(\cE(X),Z^X)$ as the minimal solution of such a backward equation (Theorem~\ref{th:2bsde}). Our last result concerns the uniqueness of time-consistent extensions and gives conditions under which~\eqref{eq:aggreg} is indeed the only possible extension of the static expectation~\eqref{eq:staticRobustRep}. In particular, we introduce the notion of local strict monotonicity to deal with the singularity of the measures (Proposition~\ref{pr:uniqueness}).

To obtain our results, we rely on methods from stochastic optimal control and the general theory of stochastic processes.
Indeed, from the point of view of dynamic programming, $\cE^\circ_t(X)$ is the value process of a control problem defined over a set of measures, and time consistency corresponds to Bellman's principle. Taking the control representation~\eqref{eq:aggreg} as our starting point allows us to consider the measures $P\in\cP$ separately in many arguments and therefore to apply standard arguments of the general theory.

The remainder of this paper is organized as follows. In Section~\ref{se:preliminaries} we detail the setting and notation. Section~\ref{se:pasting} relates time consistency to a pasting property. In Section~\ref{se:Emart} we construct the $\cE$-martingale and provide the optional sampling theorem, the decomposition, and the characterization by a 2BSDE. Section~\ref{se:uniqueness} studies the uniqueness of time-consistent extensions.

\section{Preliminaries}\label{se:preliminaries}

We fix a constant $T>0$ and let $\Omega= \{\omega\in C([0,T];\mathbbm{R}^d):\,\omega_0=0\}$ be
the canonical space of continuous paths equipped with the uniform topology.
We denote by $B$ the canonical process $B_t(\omega)=\omega_t$, by $P_0$ the Wiener measure and by
$\Fci= \{\cFci_t\}_{0\leq t\leq T}$, $\cFci_t=\sigma(B_s,\,s\leq t)$ the raw filtration generated by $B$.
As in \cite{DenisHuPeng.2010,Nutz.10Gexp, SonerTouziZhang.2010dual,SonerTouziZhang.2010bsde} we shall use the so-called strong formulation of volatility uncertainty in this paper; i.e., we consider martingale laws induced by stochastic integrals of $B$ under $P_0$. More precisely, we define $\cP_S$ to be the set of laws
\begin{equation}\label{eq:strongFormulation}
  P^\alpha := P_0 \circ (X^\alpha)^{-1}, \quad\mbox{where}\quad X^{\alpha}_t := \sideset{^{(P_0)\hspace{-7pt}}}{}{\int_0^t} \alpha_s^{1/2} dB_s,\quad t\in [0,T]
\end{equation}
and $\alpha$ ranges over all $\Fci$-progressively measurable processes with values in $\mathbb{S}^{>0}_d$ satisfying $\int_0^T |\alpha_t|\,dt<\infty$ $P_0$-a.s. Here $\mathbb{S}^{>0}_d\subset \R^{d\times d}$ denotes the set of strictly positive definite matrices and the stochastic integral in~\eqref{eq:strongFormulation} is the It\^o integral under $P_0$, constructed in $\Fci$ (cf.~\cite[p.\,97]{StroockVaradhan.79}). We remark that $\cP_S$ coincides with the set denoted by $\overline{\cP}_S$ in \cite{SonerTouziZhang.2010aggreg}.

The basic object in this paper is a nonempty set $\cP\subseteq \cP_S$ which represents the possible scenarios for the volatility. For $t\in[0,T]$, we define
$L^1_\cP(\cFci_t)$ to be the space of $\cFci_t$-measurable random variables $X$ satisfying
\[
  \|X\|_{L^1_\cP}:=\sup_{P\in \cP} \|X\|_{L^1(P)}<\infty,
\]
where $\|X\|_{L^1(P)}:=E[|X|]$. More precisely, we take equivalences classes with respect to $\cP$-quasi-sure equality so that $L^1_\cP(\cFci_t)$ becomes a Banach space. (Two functions are equal $\cP$-quasi-surely, $\cP$-q.s.\ for short, if they are equal up to a $\cP$-polar set. A set is called $\cP$-polar if it is a $P$-nullset for all $P\in\cP$.)
We also fix a nonempty subset $\cH$ of $L^1_\cP:=L^1_\cP(\cFci_T)$ whose elements play the role of financial claims. We emphasize that in applications, $\cH$ is typically smaller than $L^1_\cP$. The following is a motivating example for many of the considerations in this paper.

\begin{example}\label{ex:Gexp}
  (i)~Given real numbers $0\leq \underline{a}\leq\overline{a}<\infty$, the associated $G$-expectation (for dimension $d=1$) corresponds to the choice
  \begin{equation}\label{eq:cPforGexp}
    \cP=\big\{P^\alpha\in\cP_S:\, \underline{a}\leq \alpha \leq \overline{a}\quad P_0\times dt\mbox{-a.e.}\big\},
  \end{equation}
  cf.~\cite[Section~3]{DenisHuPeng.2010}. Here the symbol $G$ refers to the function
  \[
    G(\gamma):=\frac{1}{2}\sup_{\underline{a}\leq a \leq \overline{a}} a\gamma.
  \]
  If $X=f(B_T)$ for a sufficiently regular function $f$, then $\cE^{\circ,G}_t(X)$ is defined via the solution of the nonlinear heat equation $-\partial_t u - G(u_{xx})=0$ with boundary condition $u|_{t=T}=f$. In~\cite{Peng.07}, the mapping $\cE^{\circ,G}_t$ is extended to random variables of the form $X=f(B_{t_1},\dots,B_{t_n})$ by a stepwise evaluation of the PDE and finally to the $\|\cdot\|_{L^1_\cP}$-completion $\cH$ of the set of
  all such random variables.  For $X\in\cH$, the $G$-expectation then satisfies
  \[
    \cE^{\circ,G}_t(X)={\mathop{\esssup^P}_{P'\in \cP(\cFci_t,P)}} E^{P'}[X|\cFci_t]\quad P\as
    \quad\mbox{for all } P\in\cP,
  \]
  which is of the form~\eqref{eq:aggreg}. The space $\cH$ coincides with the $\|\cdot\|_{L^1_\cP}$-completion of $C_b(\Omega)$, the set of bounded continuous functions on $\Omega$, and is strictly smaller than $L^1_\cP$ as soon as  $\underline{a}\neq\overline{a}$.

  (ii)~The random $G$-expectation corresponds to the case where $\underline{a}$, $\overline{a}$ are random processes instead of constants and is directly constructed from a set $\cP$ of measures (cf.~\cite{Nutz.10Gexp}). In this case the space
  $\cH$ is the $\|\cdot\|_{L^1_\cP}$-completion of $\UC_b(\Omega)$, the set of bounded uniformly continuous functions on $\Omega$. If $\overline{a}$ is finite-valued and uniformly bounded, $\cH$ coincides with the space from~(i).
\end{example}

\section{Time Consistency and Pasting}\label{se:pasting}

In this section, we consider time consistency as a property of the set $\cP\subseteq \cP_S$ and obtain some auxiliary results for later use. The set $\cH\subseteq L^1_\cP$ is fixed throughout. Moreover, we let $\cT(\Fci)$ be the set of all $\Fci$-stopping times taking finitely many values; this choice is motivated by the applications in the subsequent section. However, the results of this section hold true  also if $\cT(\Fci)$ is replaced by an arbitrary set of $\Fci$-stopping times containing $\sigma\equiv0$; in particular, the set of all stopping times and the set of all deterministic times.
Given $\cA\subseteq \cFci_T$ and $P\in\cP$, we use the standard notation
\[
  \cP(\cA,P)=\{P'\in\cP:\, P'=P\mbox{ on }\cA\}.
\]
At the level of measures, time consistency can then be defined as follows.

\begin{definition}\label{def:PtimeConsistent}
  The set $\cP$ is \emph{$\Fci$-time-consistent on} $\cH$ if
  \begin{equation}\label{eq:timeConsP}
    \mathop{\esssup^P}_{P'\in\cP(\cFci_\sigma,P)} E^{P'}\bigg[ \mathop{\esssup^{P'}}_{P''\in\cP(\cFci_\tau,P')} E^{P''}[X|\cFci_\tau]\bigg|\cFci_\sigma\bigg]
    = \mathop{\esssup^P}_{P'\in\cP(\cFci_\sigma,P)} E^{P'}[X|\cFci_\sigma]\quad P\as
  \end{equation}
  for all $P\in\cP$, $X\in\cH$ and $\sigma\leq \tau$ in $\cT(\Fci)$.
\end{definition}

This property embodies the principle of dynamic programming (e.g., \cite{ElKaroui.81}). We shall relate it to the following notion of stability, also called  m-stability, fork-convexity, stability under concatenation, etc.

\begin{definition}\label{def:stableUnderPasting}
  The set $\cP$ is \emph{stable under $\Fci$-pasting} if for all $P\in\cP$, $\tau\in \cT(\Fci)$, $\Lambda\in\cFci_\tau$
  and $P_1,P_2\in\cP(\cFci_\tau,P)$, the measure $\bar{P}$ defined by
  \begin{equation}\label{eq:defStableUnderPasting}
    \bar{P}(A):=E^P \big[P_1(A|\cFci_\tau)\1_\Lambda + P_2(A|\cFci_\tau)\1_{\Lambda^c}\big],\quad A\in\cFci_T
  \end{equation}
  is again an element of $\cP$. %
\end{definition}

As $\Fci$ is the only filtration considered in this section, we shall sometimes omit the qualifier ``$\Fci$''. %

\begin{lemma}\label{le:strongMeasuresStable}
  The set $\cP_S$ is stable under pasting.
\end{lemma}

\begin{proof}
  Let $P,P_1,P_2,\tau,\Lambda,\bar{P}$ be as in Definition~\ref{def:stableUnderPasting}. Using the notation~\eqref{eq:strongFormulation}, let $\alpha,\alpha^i$ be such that $P^{\alpha}=P$ and $P^{\alpha^i}=P_i$ for $i=1,2$.
  Setting
  \begin{align*}
    \bar{\alpha}&_u(\omega):=\\
    &\1_{[\![0,\tau(X^\alpha)]\!]}(u)\alpha_u(\omega) + \1_{]\!]\tau(X^\alpha),T]\!]}(u)\Big[\alpha^1_u(\omega)\1_{\Lambda}(X^\alpha(\omega)) + \alpha^2_u(\omega)\1_{\Lambda^c}(X^\alpha(\omega))\Big],
  \end{align*}
  we have $\bar{P}=P^{\bar{\alpha}}\in\cP_S$ by the arguments in \cite[Appendix]{SonerTouziZhang.2010dual}.
\end{proof}

The previous proof also shows that the set appearing in~\eqref{eq:cPforGexp} is stable under pasting. The following result is classical.

\begin{lemma}\label{le:increasingSequence}
  Let $\tau\in\cT(\Fci)$, $X\in L^1_\cP$ and $P\in\cP$. If $\cP$ is stable under pasting, then there exists a sequence $P_n\in\cP(\cFci_\tau,P)$ such that
  \[
    {\mathop{\esssup^P}_{P'\in \cP(\cFci_\tau,P)}} E^{P'}[X|\cFci_\tau]=\lim_{n\to\infty} E^{P_n}[X|\cFci_\tau]\quad P\as,
  \]
  where the limit is increasing $P$-a.s.
\end{lemma}

\begin{proof}
  It suffices to show that the family $\{E^{P'}[X|\cFci_\tau]:\, P'\in\cP(\cFci_\tau,P)\}$ is
  $P$-a.s.\ upward filtering (cf.\ \cite[Proposition~VI-1-1]{Neveu.75}). Given $P_1,P_2\in\cP(\cFci_\tau,P)$, we set
  \[
    \Lambda:=\big\{E^{P_1}[X|\cFci_\tau]>E^{P_2}[X|\cFci_\tau]\big\}\in\cFci_\tau
  \]
  and define
  $\bar{P}(A):=E^P\big[P_1(A|\cFci_\tau)\1_\Lambda + P_2(A|\cFci_\tau)\1_{\Lambda^c}\big]$.
  Then $\bar{P}=P$ on $\cFci_\tau$ and $\bar{P}\in\cP$ by the stability. Moreover,
  \[
    E^{\bar{P}}[X|\cFci_\tau]=E^{P_1}[X|\cFci_\tau]\vee E^{P_2}[X|\cFci_\tau]\quad P\as,
  \]
  showing that the family is upward filtering.
\end{proof}

To relate time consistency to stability under pasting, we introduce the following closedness property.

\begin{definition}\label{def:maxChosen}
  We say that $\cP$ is \emph{maximally chosen for $\cH$} if $\cP$ contains all $P\in\cP_S$ satisfying
  $E^P[X]\leq \sup_{P'\in\cP}E^{P'}[X]$ for all $X\in \cH$.
\end{definition}

If $\cP$ is dominated by a reference probability $P_*$, then $\cP$ can be identified with a subset of $L^1(P_*)$  by the Radon-Nikodym theorem. If furthermore
$\cH=L^\infty(P_*)$, the Hahn-Banach theorem implies that $\cP$ is maximally chosen if and only if $\cP$ is convex and closed for
weak topology of $L^1(P_*)$. Along these lines, the following result can be seen as a generalization of \cite[Theorem~12]{Delbaen.06}; in fact, we merely replace functional-analytic arguments by algebraic ones.

\begin{proposition}\label{pr:pastingAndTimeconsistency}
  With respect to the filtration $\Fci$, we have:
  \begin{enumerate}[topsep=3pt, partopsep=0pt, itemsep=1pt,parsep=2pt]
    \item If $\cP$ is stable under pasting, then $\cP$ is time-consistent on $L^1_\cP$.
    \item If $\cP$ is time-consistent on $\cH$ and maximally chosen for $\cH$, then $\cP$ is stable under pasting.
  \end{enumerate}
\end{proposition}

\begin{proof}
  (i)~This implication is standard; we provide the argument for later reference. The inequality ``$\geq$'' in~\eqref{eq:timeConsP} follows by considering
  $P'':=P'$ on the left hand side. To see the converse inequality, fix an arbitrary $P\in\cP$ and choose a sequence $P_n\in\cP(\cFci_\tau,P)\subseteq\cP(\cFci_\sigma,P)$ as in
  Lemma~\ref{le:increasingSequence}. Then monotone convergence yields
   \begin{align*}
    E^P\bigg[{\mathop{\esssup^P}_{P'\in \cP(\cFci_\tau,P)}} E^{P'}[X|\cFci_\tau]\bigg|\cFci_\sigma\bigg]
    & = \lim_{n\to\infty} E^{P_n}[X|\cFci_\sigma] \\
    & \leq {\mathop{\esssup^P}_{P'\in \cP(\cFci_\sigma,P)}} E^{P'}[X|\cFci_\sigma]\quad P\as
  \end{align*}

  (ii)~Let $\cP$ be time-consistent and let $P,P_1,P_2,\tau,\Lambda,\bar{P}$ be as in Definition~\ref{def:stableUnderPasting}. For any $X\in \cH$, we have
  \begin{align*}
  E^{\bar{P}}[X]
       &=E^P\Big[E^{P_1}[X|\cFci_\tau]\1_\Lambda + E^{P_2}[X|\cFci_\tau]\1_{\Lambda^c}\Big] \\
       &\leq E^{P}\bigg[ \mathop{\esssup^P}_{P''\in\cP(\cFci_\tau,P)} E^{P''}[X|\cFci_\tau]\bigg] \\
       &\leq \sup_{P'\in \cP} E^{P'}\bigg[ \mathop{\esssup^{P'}}_{P''\in\cP(\cFci_\tau,P')} E^{P''}[X|\cFci_\tau]\bigg]\\
       & = \sup_{P'\in\cP} E^{P'}[X],
  \end{align*}
  where the last equality uses~\eqref{eq:timeConsP} with $\sigma\equiv0$.
  Since $\cP$ is maximally chosen and $\bar{P}\in \cP_S$ by Lemma~\ref{le:strongMeasuresStable}, we conclude that $\bar{P}\in\cP$.
\end{proof}

\section{$\cE$-Martingales}\label{se:Emart}

As discussed in the introduction, our starting point in this section is a given family $\{\cE^{\circ}_t(X),\,t\in[0,T]\}$ of random variables which will serve as a raw version of the $\cE$-martingale to be constructed. We recall that the sets $\cP\subseteq \cP_S$ and $\cH\subseteq L^1_\cP$ are fixed.

\begin{assumption}\label{as:pastingAndAggreg}
  Throughout Section~\ref{se:Emart}, we assume that
  \begin{enumerate}[topsep=3pt, partopsep=0pt, itemsep=1pt,parsep=2pt]
   \item for all $X\in\cH$ and $t\in[0,T]$, there exists an $\cFci_t$-measurable random variable $\cE^\circ_t(X)$ such that
    \begin{equation}\label{eq:ErawAggreg}
      \cE^\circ_t(X)={\mathop{\esssup^P}_{P'\in \cP(\cFci_t,P)}} E^{P'}[X|\cFci_t]\quad P\as
      \quad\mbox{for all } P\in\cP.
    \end{equation}
   \item the set $\cP$ is stable under $\Fci$-pasting.
  \end{enumerate}
\end{assumption}

\pagebreak[2]

The first assumption was discussed in the introduction; cf.~\eqref{eq:aggreg}. With the motivating Example~\ref{ex:Gexp} in mind, we ask for~\eqref{eq:ErawAggreg} to hold at deterministic times rather than at stopping times.
 The second assumption is clearly motivated by Proposition~\ref{pr:pastingAndTimeconsistency}(ii), and Proposition~\ref{pr:pastingAndTimeconsistency}(i) shows that $\cP$ is time-consistent in the sense of Definition~\ref{def:PtimeConsistent}. (We could assume the latter property directly, but stability under pasting is more suitable for applications.) In particular, we have
\begin{equation}\label{eq:timeConsEraw}
    \cE^\circ_s(X) = \mathop{\esssup^P}_{P'\in\cP(\cFci_s,P)} E^{P'} [ \cE^\circ_t(X) |\cFci_s]\quad P\as\quad \mbox{for all }P\in\cP,
\end{equation}
$0\leq s\leq t \leq T$ and $X\in\cH$. If we assume that $\cE^\circ_t(X)$ is again an element of the domain $\cH$, this amounts to $\{\cE^\circ_t\}$ being time-consistent (at deterministic times) in the sense that the semigroup property $\cE^\circ_s\circ \cE^\circ_t=\cE_s^\circ$ is satisfied.
However, $\cE^\circ_t(X)$ need not be in $\cH$ in general; e.g., for certain random \mbox{$G$-expectations}. Inspired by the theory of viscosity solutions, we introduce the following extended notion of time consistency, which is clearly implied by~\eqref{eq:timeConsEraw}.

\begin{definition}\label{def:EtimeConsistent}
  A family $(\E_t)_{0\leq t\leq T}$ of mappings $\E_t: \cH \to L^1_\cP(\cFci_t)$ is called \emph{$\Fci$-time-consistent at deterministic times} if for all $0\leq s\leq t \leq T$ and $X\in\cH$,
  \[
    \E_s(X)\leq \,(\geq)\,\E_s(\varphi) \quad\mbox{for all }\varphi\in L^1_\cP(\cFci_t)\cap \cH \mbox{ such that }\E_t(X)\leq \,(\geq)\, \varphi.
  \]
\end{definition}

One can give a similar definition for stopping times taking countably many values. (Note that $\E_\tau(X)$ is not necessarily well defined for a general stopping time $\tau$.)

\begin{remark}
  If Assumption~\ref{as:pastingAndAggreg} is weakened by requiring $\cP$ to be stable only under $\Fci$-pastings at deterministic times (i.e., Definition~\ref{def:PtimeConsistent} holds with $\cT(\Fci)$ replaced by the set of deterministic times), then all results in this section remain true with the same proofs, except for Theorem~\ref{th:DPPstop}, Lemma~\ref{le:classD} and the last statement in Theorem~\ref{th:2bsde}.
\end{remark}

\subsection{Construction of the $\cE$-Martingale}

Our first task is to turn the collection $\{\cE^\circ_t(X),\,t\in[0,T]\}$ of random variables into a reasonable stochastic process. As usual, this requires an extension of the filtration. We denote by
\[
  \F^+=\{\cF^+_t\}_{0\leq t\leq T},\quad \cF^+_t:= \cFci_{t+}
\]
the minimal right continuous filtration containing $\Fci$; i.e., $\cFci_{t+}:=\bigcap_{s>t}\cFci_s$ for $0\leq t<T$ and $\cFci_{T+}:=\cFci_T$. We augment $\F^+$ by the collection
$\cN^{\cP}$ of $(\cP,\cFci_T)$-polar sets to obtain the filtration
\[
  \hF=\{\hcF_t\}_{0\leq t\leq T},\quad \hcF_t:= \cFci_{t+}\vee\cN^{\cP}.
\]
Then $\hF$ is right continuous and a natural analogue of the ``usual augmentation'' that is standard in the case where a reference probability is given. More precisely, if $\cP$ is dominated by some probability measure, then one can find a minimal dominating measure $P_*$ (such that every $\cP$-polar set is a $P_*$-nullset) and then $\hF$ coincides with the $P_*$-augmentation of $\F^+$.
We remark that $\hF$ is in general strictly smaller than the $\cP$-universal augmentation $\bigcap_{P\in\cP} \overline{\Fci}^P$, which seems to be too large for our purposes. Here $\overline{\Fci}^P$ denotes the $P$-augmentation of $\Fci$.

Since $\F$ and $\hF^+$ differ only by $\cP$-polar sets, they can be identified for most purposes; note in particular that $\cF_T=\cF^+_{T}=\cFci_{T}$ $\cP$-q.s.
We also recall the following result (e.g., \cite[Theorem~1.5]{JacodYor.77}, \cite[Lemma~8.2]{SonerTouziZhang.2010aggreg}), which shows that
$\hF$ and $\Fci$ differ only by $P$-nullsets for each $P\in\cP$.

\begin{lemma}\label{le:MRPandVersions}
  Let $P\in\cP$. Then $\overline{\Fci}^P$ is right continuous and
  in particular contains $\hF$. Moreover,
  $(P,B)$ has the predictable representation property; i.e., for any right continuous $(\overline{\Fci}^P,P)$-local martingale $M$ there exists an $\overline{\Fci}^P$-predictable process $Z$ such that
  $M=M_0+{}^{(P)\hspace{-5pt}}\int Z\,dB$, $P$-a.s.
\end{lemma}

\begin{proof}
  We sketch the argument for the convenience of the reader. We define a predictable process
 $\hat{a}_t=d\br{B}_t/dt$ taking values in $\mathbb{S}^{>0}_d$ $P\times dt$-a.e., note that $(\hat{a})^{-1/2}$ is square-integrable for $B$ by its very definition,
  and consider $W_t:={}^{(P)\hspace{-5pt}}\int_0^t (\hat{a}_u)^{-1/2}\,dB_u$. Let $\F^{W}$ be the raw filtration generated by $W$. Since $W$ is a $P$-Brownian motion by L\'evy's characterization, the $P$-augmentation $\overline{\F^W}^P$ is right continuous and $W$ has the representation property. Moreover, as $P\in\cP_S$, \cite[Lemma~8.1]{SonerTouziZhang.2010aggreg} yields that $\overline{\F^W}^P=\overline{\Fci}^P$. Thus $\overline{\Fci}^P$ is also right continuous and $B$ has the representation property since any integral of $W$ is also an integral of $B$.
\end{proof}

We deduce from Lemma~\ref{le:MRPandVersions} that for $P\in\cP$, any (local) $(\Fci,P)$-martingale is a (local) $(\hF,P)$-martingale. In particular, this applies to the canonical process $B$.
Note that Lemma~\ref{le:MRPandVersions} does not imply that $\hF$ and $\Fci$ coincide up to $\cP$-polar sets. E.g., consider the set
\begin{equation}\label{eq:SetForCouterexamples}
  A:=\Big\{\limsup_{t\to0} t^{-1}\br{B}_t = \liminf_{t\to 0} t^{-1}\br{B}_t =1 \Big\}\in\cFci_{0+}.
\end{equation}
Then the lemma asserts that $P(A)\in\{0,1\}$ for all $P\in\cP$, but not that this number is the same for all $P$.
Indeed, $P^\alpha(A) =1$ for $\alpha\equiv1$ but $P^\alpha(A)=0$ for $\alpha\equiv2$.

We can now state the existence and uniqueness of the stochastic process derived from $\{\cE^\circ_t(X),\,t\in[0,T]\}$. For brevity, we shall say that $Y$ is an $(\hF,\cP)$-supermartingale if $Y$ is an $(\hF,P)$-supermartingale for all $P\in\cP$; analogous notation will be used in similar situations.

\pagebreak[4]

\begin{proposition}\label{pr:extension}
  Let $X\in\cH$. There exists an $\hF$-optional process $(Y_t)_{0\leq t\leq T}$ such that
  all paths of $Y$ are c\`adl\`ag and
  \begin{enumerate}[topsep=3pt, partopsep=0pt, itemsep=1pt,parsep=2pt]
    \item $Y$ is the minimal $(\hF,\cP)$-supermartingale with $Y_T=X$; i.e., if $S$ is a c\`adl\`ag
    $(\hF,\cP)$-supermartingale with $S_T=X$, then $S\geq Y$ up to a $\cP$-polar set.
    \item $Y_t=\cE^\circ_{t+}(X):=\lim_{r\downarrow t} \cE^\circ_r(X)$ $\cP$-q.s.\ for all $0\leq t<T$, and $Y_T=X$.
    \item $Y$ has the representation
     \begin{equation}\label{eq:Yesssup}
      Y_t = \mathop{\esssup^P}_{P'\in \cP(\hcF_t,P)} E^{P'}[X|\hcF_t]\quad P\as\quad \mbox{for all } P\in\cP.
     \end{equation}
  \end{enumerate}
  Any of the properties (i),(ii),(iii) characterizes $Y$ uniquely up to $\cP$-polar sets. The process $Y$ is denoted by $\cE(X)$ and called the (c\`adl\`ag) \emph{$\cE$-martingale} associated with $X$.
\end{proposition}

\begin{proof}
  We choose and fix representatives for the classes $\cE^\circ_t(X)\in L^1_\cP(\cFci_t)$ and define the $\R\cup\{\pm\infty\}$-valued process $Y$ by
  \[
    Y_t(\omega):=\limsup_{r\in (t,T]\cap\Q,\; r\to t} \cE^\circ_r(X)(\omega)\quad\mbox{for }0\leq t<T \quad\mbox{and}\quad Y_T(\omega):=X(\omega)
  \]
  for all $\omega\in\Omega$. Since each $\cE^\circ_r(X)$ is $\cFci_r$-measurable, $Y$ is adapted to $\F^+$ and in particular to $\hF$.
  Let $N$ be the set of $\omega\in\Omega$ for which there exists $t\in [0,T)$ such that $\lim_{r\in (t,T]\cap\Q,\; r\to t} \cE^\circ_r(X)(\omega)$ does not exist as a finite real number.
  For any $P\in\cP$,~\eqref{eq:timeConsEraw} implies the $(\Fci,P)$-supermartingale property
  \[
    \cE^\circ_s(X)\geq E^P[\cE^\circ_t(X)|\cFci_s]\quad P\as,\quad 0\leq s\leq t\leq T.
  \]
  Thus the standard modification argument for supermartingales (see \cite[Theorem~VI.2]{DellacherieMeyer.82}) yields that
  $P(N)=0$. As this holds for all $P\in\cP$, the set $N$ is $\cP$-polar and thus $N\in\hcF_0$. We redefine $Y:=0$ on $N$. Then all paths of $Y$ are finite-valued and c\`adl\`ag. Moreover, the resulting process is $\hF$-adapted and therefore $\hF$-optional by the c\`adl\`ag property.
  Of course, redefining $Y$ on $N$ does not affect the $P$-almost sure properties of $Y$. In particular,~\cite[Theorem~VI.2]{DellacherieMeyer.82} shows that $Y$ is an $(\hF,P)$-supermartingale.

  Let $P'\in\cP(\hcF_t,P)$.
  Using the above observation with $P'$ instead of $P$, we also have that $Y$ is an $(\hF,P')$-supermartingale. As $X=Y_T$, this yields that $E^{P'}[X|\hcF_t]=E^{P'}[Y_T|\hcF_t]\leq Y_t$ $P'$-a.s., and also $P$-a.s.\ because $P'=P$ on $\hcF_t$. Since $P'\in\cP(\hcF_t,P)$ was arbitrary, we conclude that
  \begin{equation}\label{eq:proofExtensionIneq}
    Y_t\geq \mathop{\esssup^P}_{P'\in \cP(\hcF_t,P)} E^{P'}[X|\hcF_t] \quad P\as
  \end{equation}
  To see the converse inequality, consider a strictly decreasing sequence $t_n\downarrow t$ of rationals.
  Then $\cE^\circ_{t_n}(X)\to Y_t$ $P$-a.s.\ by the definition of $Y_t$, but
  as $E^P[\cE^\circ_{t_n}(X)]\leq \cE^\circ_{0}(X)<\infty$, the backward supermartingale convergence theorem \cite[Theorem~V.30]{DellacherieMeyer.82} shows that this convergence holds also in $L^1(P)$ and
  hence
  \begin{equation}\label{eq:proofExtensionConv}
    Y_t=\lim_{n\to\infty} E^P[\cE^\circ_{t_n}(X)|\hcF_t]\quad \mbox{in }L^1(P)\mbox{ and } P\as
  \end{equation}
  Here the convergence in $L^1(P)$ holds by the $L^1(P)$-continuity of $E^P[\,\cdot\,|\hcF_t]$ and then the convergence $P$-a.s.\
  follows since the sequence on the right hand side is monotone by the supermartingale property.
  For fixed $n$, let $P^n_k\in\cP(\cFci_{t_n},P)$ be a sequence as in Lemma~\ref{le:increasingSequence}. Then monotone convergence yields
   \begin{align*}
    E^P[\cE^\circ_{t_n}(X)|\hcF_t]
    &= E^P\bigg[\mathop{\esssup^P}_{P'\in \cP(\cFci_{t_n},P)} E^{P'}[X|\cFci_{t_n}]\bigg|\hcF_t\bigg] \\
    & = \lim_{k\to\infty} E^{P^n_k}[X|\hcF_t] \\
    & \leq \mathop{\esssup^P}_{P'\in \cP(\hcF_t,P)} E^{P'}[X|\hcF_t]\quad P\as,
  \end{align*}
  since $P^n_k\in\cP(\hcF_t,P)$ for all $k$ and $n$;
  indeed, we have $P^n_k\in\cP(\cFci_{t_n},P)$ and $\cP(\cFci_{t_n},P)\subseteq \cP(\cFci_{t+},P)$ since $t_n>t$, moreover,
  $\cP(\cFci_{t+},P)=\cP(\hcF_t,P)$ since $\cFci_{t+}$ and $\hcF_t$ coincide up to $\cP$-polar sets.
  In view of~\eqref{eq:proofExtensionConv}, the inequality converse to~\eqref{eq:proofExtensionIneq} follows and~(iii) is proved.

  To see the minimality property in~(i), let $S$ be an $(\hF,\cP)$-supermartingale with $S_T=X$. Exactly as in~\eqref{eq:proofExtensionIneq}, we deduce that
  \[
     S_t\geq \mathop{\esssup^P}_{P'\in \cP(\hcF_t,P)} E^{P'}[X|\hcF_t] \quad P\as\quad\mbox{for all }P\in\cP.
  \]
  By~(iii) the right hand side is $P$-a.s.\ equal to $Y_t$. Hence $S_t\geq Y_t$ $\cP$-q.s.\ for all $t$ and
  $S\geq Y$ $\cP$-q.s.\ when $S$ is c\`adl\`ag

  Finally, if $Y$ and $Y'$ are processes satisfying (i) or (ii) or (iii), then they are $P$-modifications of each other for all $P\in\cP$ and thus coincide up to a $\cP$-polar set as soon as they are c\`adl\`ag.
\end{proof}

One can ask whether $\cE(X)$ is a \emph{$\cP$-modification} of
$\{\cE^\circ_t(X),\,t\in[0,T]\}$; i.e., whether
\[
  \cE_t(X)=\cE^\circ_t(X)\quad\cP\qs\quad\mbox{for all }0\leq t\leq T.
\]
It is easy to see that $\cE(X)$ is a $\cP$-modification as soon as there exists \emph{some} c\`adl\`ag $\cP$-modification of the family $\{\cE^\circ_t(X),\,t\in[0,T]\}$, and this is the case if and only if $t\mapsto E^P[\cE^\circ_t(X)]$ is right continuous for all $P\in\cP$.
We also remark that Lemma~\ref{le:MRPandVersions} and the argument given for~\eqref{eq:proofExtensionIneq} yield
\begin{equation}\label{eq:ineqModification}
  \cE_t(X)\leq \cE^\circ_t(X)\quad\cP\qs\quad\mbox{for all }0\leq t\leq T
\end{equation}
and so the question is only whether the converse inequality holds true as well.
The answer is positive in several important cases; e.g., for the $G$-expectation when $X$ is sufficiently regular \cite[Theorem~5.3]{Song.10} and the sublinear expectation generated by a controlled stochastic differential equation \cite[Theorem~5.1]{Nutz.11}. The proof of the latter result yields a general technique to approach this problem in a given example.
However, the following (admittedly degenerate) example shows that the answer is negative in a very general case; this reflects the fact that the set $\cP(\hcF_t,P)$ in the representation~\eqref{eq:Yesssup} is smaller than the set $\cP(\cFci_t,P)$ in~\eqref{eq:ErawAggreg}.

\begin{example}\label{ex:counterexModification}
  We shall consider a $G$-expectation defined on a set of irregular random variables. Let $\underline{a}=1$, $\overline{a}=2$ and let $\cP$ be as in~\eqref{eq:cPforGexp}. We take $\cH=L^1_{\cP}(\cF^\circ_{0+})$ and define
  \[
  \cE^{\circ}_t(X):=
  \begin{cases}
    \sup_{P\in\cP} E^P[X], & t=0,  \\
    X, & 0<t\leq T
  \end{cases}
  \]
  for $X\in\cH$. Then $\{\cE^{\circ}_t\}$ trivially satisfies~\eqref{eq:ErawAggreg} since $X$ is $\cFci_t$-measurable for all $t>0$. As noted after Lemma~\ref{le:strongMeasuresStable}, the second part of Assumption~\ref{as:pastingAndAggreg} is also satisfied. Moreover, the c\`adl\`ag $\cE$-martingale is given by
  \[
    \cE_t(X)=X,\quad t\in [0,T].
  \]
  Consider $X:=\1_A$, where $A$ is defined as in~\eqref{eq:SetForCouterexamples}. Then $\cE^{\circ}_0(X)=1$ and $\cE_0(X)=\1_A$ are not equal $P^2$-a.s.\ (i.e., the measure $P^\alpha$ for $\alpha\equiv 2$). In fact, there is no c\`adl\`ag $\cP$-modification since $\{\cE^{\circ}_t(X)\}$ coincides $P^2$-a.s.\ with the deterministic function $t\mapsto \1_{\{0\}}(t)$.
\end{example}

We remark that the phenomenon appearing in the previous example is due to the presence of singular measures rather than the fact that $\cP$ is not dominated. In fact, one can give a similar example involving only two measures.

Finally, let us mention that the situation is quite different if we assume that the given sublinear expectation is already placed in the larger filtration $\F$ (i.e., Assumption~\ref{as:pastingAndAggreg} holds with $\Fci$ replaced by $\F$), which would be in line with the paradigm of the ``usual assumptions'' in standard stochastic analysis. In this case, the arguments in the proof of Proposition~\ref{pr:extension} show that $\cE(X)$ is always a $\cP$-modification. This result is neat, but not very useful, since the examples are typically constructed in $\Fci$.

\subsection{Stopping Times}

The direct construction of $G$-expectations at stopping times is an unsolved problem. Indeed, stopping times are typically fairly irregular functions and it is unclear how to deal with this in the existing constructions (see also~\cite{LiPeng.09}). On the other hand, we can easily evaluate the c\`adl\`ag process $\cE(X)$ at a stopping time $\tau$ and therefore define the corresponding sublinear expectation at $\tau$. In particular, this leads to a definition of $G$-expectations at general stopping times. We show in this section that the resulting random variable $\cE_\tau(X)$ indeed has the expected properties and that the time consistency extends to arbitrary $\F$-stopping times; in other words, we prove an optional sampling theorem for $\cE$-martingales. Besides the obvious theoretical interest, the study of $\cE(X)$ at stopping times will allow us to verify integrability conditions of the type ``class~(D)''; cf.\ Lemma~\ref{le:classD} below.
We start by explaining the relations between the stopping times of the different filtrations.

\begin{lemma}\label{le:stopFields}
  (i)~Let $P\in\cP$ and let $\tau$ be an $\hF$-stopping time taking countably many values. Then there exists
  an $\Fci$-stopping time $\tau^\circ$ (depending on $P$) such that $\tau=\tau^\circ$ $P$-a.s. Moreover, for any such $\tau^\circ$, the $\sigma$-fields $\hcF_\tau$ and $\cFci_{\tau^\circ}$ differ only by $P$-nullsets.

  (ii)~Let $\tau$ be an $\hF$-stopping time. Then there exists
  an $\F^+$-stopping time $\tau^+$ such that $\tau=\tau^+$ $\cP$-q.s. Moreover, for any such $\tau^+$,  the $\sigma$-fields $\hcF_\tau$ and $\cF^+_{\tau^+}$ differ only by $\cP$-polar sets.
\end{lemma}

\begin{proof}
  (i)~Note that $\tau$ is of the form $\tau=\sum_i t_i \1_{\Lambda_i}$ for $\Lambda_i=\{\tau=t_i\}\in \hcF_{t_i}$ forming a partition of $\Omega$. Since $\hF\subseteq \overline{\Fci}^P$ by Lemma~\ref{le:MRPandVersions},
  we can find $\Lambda^\circ_i\in \cFci_{t_i}$ such that $\Lambda_i=\Lambda^\circ_i$ $P$-a.s.\ and the first assertion follows by taking
  \[
    \tau^\circ :=T\1_{(\cup_i \Lambda^\circ_i)^c}+ \sum_i t_i \1_{\Lambda^\circ_i}.
  \]

  Let $A\in\hcF_\tau$. By the first part,
  there exists an $\Fci$-stopping time $(\tau_A)^\circ$ such that $(\tau_A)^\circ=\tau_A:=\tau\1_A + T \1_{A^c}$ $P$-a.s. Moreover, we choose
  $A'\in\cFci_T$ such that $A=A'$ $P$-a.s. Then
  \[
    A^\circ :=\big( A'\cap\{ \tau^\circ =T\}\big) \cup \{(\tau_A)^\circ = \tau^\circ < T\}
  \]
  satisfies $A^\circ \in \cFci_{\tau^\circ}$ and $A=A^\circ$ $P$-a.s. A similar but simpler argument shows that
  for given $\Lambda\in\cFci_{\tau^\circ}$ we can find $\Lambda'\in\hcF_\tau$ such that $\Lambda=\Lambda'$ $P$-a.s.

  (ii)~If $\tau$ is an $\hF$- (resp. $\F^+$-) stopping time, we can find $\tau^n$ taking countably many values such that $\tau^n$ decreases to $\tau$ and since $\hF$ ($\F^+$) is right continuous, $\hcF_{\tau^n}$ ($\cF^+_{\tau^n}$) decreases to $\hcF_{\tau}$ ($\cF^+_{\tau}$).
  As a result, we may assume without loss of generality that $\tau$ takes countably many values.

  Let $\tau=\sum_i t_i \1_{\Lambda_i}$, where $\Lambda_i\in \hcF_{t_i}$. The definition of $\hF$ shows that there exist $\Lambda^+_i\in \cF^+_{t_i}$ such that $\Lambda_i=\Lambda^+_i$ $\cP$-q.s.\ and the first part follows. The proof of the second part is as in (i); we now have quasi-sure instead of almost-sure relations.
\end{proof}

If $\sigma$ is a stopping time taking finitely many values $(t_i)_{1\leq i\leq N}$, we can define $\cE^\circ_\sigma(X) :=\sum_{i=1}^N \cE^\circ_{t_i}(X) \1_{\{\sigma=t_i\}}$. We have the following generalization of~\eqref{eq:ErawAggreg}.

\begin{lemma}\label{le:rawStopping}
  Let $\sigma$ be an $\Fci$-stopping time taking finitely many values.
  Then
  \[
    \cE^\circ_\sigma(X) = \mathop{\esssup^P}_{P'\in \cP(\cFci_\sigma,P)} E^{P'}[X|\cFci_\sigma]\quad \quad P\as\quad\mbox{for all } P\in\cP.
  \]
\end{lemma}

\begin{proof} Let $P\in\cP$ and $Y^\circ_t:=\cE^\circ_t(X)$. Moreover, let $(t_i)_{1\leq i\leq N}$ be the values of $\sigma$ and $\Lambda_i:=\{\sigma=t_i\}\in\cFci_{t_i}$.

  (i)~We first prove the inequality ``$\geq$''. Given $P'\in \cP$, it follows from~\eqref{eq:timeConsEraw} that
  $\{Y^\circ_{t_i}\}_{1\leq i\leq N}$ is a $P'$-supermartingale in $(\cFci_{t_i})_{1\leq i\leq N}$ and so the (discrete-time) optional sampling theorem~\cite[Theorem~V.11]{DellacherieMeyer.82} implies $Y^\circ_\sigma\geq E^{P'}[X|\cFci_\sigma]$ $P'$-a.s. In particular, this also holds $P$-a.s.\ for all $P'\in\cP(\cFci_\sigma,P)$, hence the claim follows.

  (ii)~We now show the inequality ``$\leq$''. Note that $\sigma=\sum_{i=1}^N t_i\1_{\Lambda_i}$ and that $(\Lambda_i)_{1\leq i\leq N}$ form an $\cFci_\sigma$-measurable partition of $\Omega$. It suffices to show that
  \[
    Y^\circ_{t_i}\1_{\Lambda_i} \leq \mathop{\esssup^P}_{P'\in \cP(\cFci_\sigma,P)} E^{P'}[X|\cFci_\sigma]\1_{\Lambda_i}\quad P\as\quad\mbox{for }1\leq i \leq N.
  \]
  In the sequel, we fix $i$ and show that for each $P'\in\cP(\cFci_{t_i},P)$ there exists $\bar{P}\in\cP(\cFci_\sigma,P)$ such that
  \begin{equation}\label{eq:pastingReq}
    \bar{P}(A\cap \Lambda_i)=P'(A\cap \Lambda_i)\quad\mbox{for all } A\in \cFci_T.
  \end{equation}
  In view of~\eqref{eq:ErawAggreg} and $E^{P'}[X|\cFci_\sigma]\1_{\Lambda_i}=E^{P'}[X|\cFci_{t_i}]\1_{\Lambda_i}$ $P'$-a.s.,
  it will then follow that
  \[
    Y^\circ_{t_i}\1_{\Lambda_i}
    = \mathop{\esssup^P}_{P'\in \cP(\cFci_{t_i},P)} E^{P'}[X\1_{\Lambda_i}|\cFci_{t_i}]
    \leq \mathop{\esssup^P}_{\bar{P}\in \cP(\cFci_\sigma,P)} E^{\bar{P}}[X\1_{\Lambda_i}|\cFci_{\sigma}]\quad P\as
  \]
  as claimed. Indeed, given $P'\in\cP(\cFci_{t_i},P)$, we define
  \begin{equation}\label{eq:pastingStopProof}
    \bar{P}(A):=P'(A\cap \Lambda_i)+ P(A\setminus \Lambda_i),\quad A\in\cFci_T,
  \end{equation}
  then~\eqref{eq:pastingReq} is obviously satisfied.
  If $\Lambda\in\cFci_\sigma$, then $\Lambda\cap \Lambda_i=\Lambda\cap \{\sigma=t_i\}\in\cFci_{t_i}$ and $P'\in\cP(\cFci_{t_i},P)$ yields
  $P'(\Lambda\cap \Lambda_i)=P(\Lambda\cap \Lambda_i)$. Hence $\bar{P}=P$ on $\cFci_\sigma$.
  Moreover, we observe that~\eqref{eq:pastingStopProof} can be stated as
  \[
    \bar{P}(A)=E^P\big[P'(A|\cFci_{t_i})\1_{\Lambda_i} + P(A|\cFci_{t_i})\1_{\Lambda_i^c}\big],\quad A\in \cFci_T,
  \]
  which is a special case of the pasting~\eqref{eq:defStableUnderPasting} applied with $P_2:=P$. Hence
  $\bar{P}\in\cP$ by Assumption~\ref{as:pastingAndAggreg} and we have $\bar{P}\in\cP(\cFci_\sigma,P)$ as desired.
\end{proof}

For the next result, we recall that stability under pasting refers to stopping times with finitely many values rather than general ones (Definition~\ref{def:stableUnderPasting}).

\begin{lemma}\label{le:pastingInF}
  The set $\cP$ is stable under $\F$-pasting.
\end{lemma}

\begin{proof}
  Let $\tau\in \cT(\F)$, then $\tau$ is of the form
  \[
    \tau=\sum_i t_i \1_{\Lambda_i},\quad \Lambda_i:=\{\tau=t_i\}\in \hcF_{t_i},
  \]
  where $t_i\in[0,T]$ are distinct and the sets $\Lambda_i$ form a partition of $\Omega$.
  Moreover, let $\Lambda\in \cF_\tau$ and $P_1,P_2\in\cP(\cF_\tau,P)$, then we have to show that the measure $E^P \big[P_1(\,\cdot\,|\cF_\tau)\1_\Lambda + P_2(\,\cdot\,|\cF_\tau)\1_{\Lambda^c}\big]$ is an element of $\cP$.

  (i) We start by proving that for any $A\in\cF_\tau$ there exists $A'\in\cFci_T\cap \cF_\tau$ such that $A=A'$ holds $\cP(\cF_\tau,P)$-q.s. Consider the disjoint union
  \[
    A = \bigcup_i (A\cap \Lambda_i).
  \]
  Here $A\cap \Lambda_i\in\cF_{t_i}$ since $A\in \cF_\tau$. As $\hF\subseteq \overline{\Fci}^P$ by Lemma~\ref{le:MRPandVersions},
  there exist a set $A_i\in \cF^\circ_{t_i}$ and a $P$-nullset $N_i$, disjoint from $A_i$, such that
  \begin{equation}\label{eq:constrCompletion}
    A\cap \Lambda_i = A_i \cup N_i.
  \end{equation}
  (It is \emph{not} necessary to subtract another nullset on the right hand side.) We define $A':=\cup_i A_i$, then $A'\in\cFci_T$ and clearly $A=A'$ $P$-a.s. Let us check that the latter also holds $\cP(\cF_\tau,P)$-q.s. For this, it suffices to show that $A'\in\cF_\tau$. Indeed, by the construction of~\eqref{eq:constrCompletion},
  \[
    A_i\cap \{\tau=t_j\}=
    \begin{cases}
      A_i\in\cFci_{t_i}\subseteq\cF_{t_i}, & i=j, \\
      \emptyset \in \cF_{t_j}, & j\neq i;
    \end{cases}
  \]
  i.e., each set $A_i$ is in $\cF_\tau$. Hence, $A'\in\cF_\tau$, which completes the proof of~(i).

  For later use, we define
  the $\Fci$-stopping time
  \[
    (\tau_A)^\circ:= T\1_{(A')^c} + \sum_i t_i\1_{A_i}
  \]
  and note that $(\tau_A)^\circ = \tau_A$ holds $\cP(\cF_\tau,P)$-q.s.

  (ii) Using the previous construction for $A=\Omega$, we see in particular that there exist $\Lambda^\circ_i\in\cFci_{t_i}$ such that $\Lambda^\circ_i=\Lambda_i$ holds $\cP(\cF_\tau,P)$-q.s. We also define the $\Fci$-stopping time
  \[
    \tau^\circ:= T\1_{(\cup_i \Lambda^\circ_i)^c} + \sum_i t_i\1_{\Lambda^\circ_i}
  \]
  which $\cP(\cF_\tau,P)$-q.s.\ satisfies $\tau^\circ=\tau$.

  (iii) We can now show that $\cFci_{\tau^\circ}$ and $\cF_{\tau}$ may be identified (when $P, P_1,P_2$ are fixed). Indeed, if $A\in\cF_\tau$, we let $A'$ be as in~(i) and set
  \[
    A^\circ :=\big( A'\cap\{ \tau^\circ =T\}\big) \cup \{(\tau_A)^\circ = \tau^\circ < T\}.
  \]
  Then $A^\circ\in \cFci_{\tau^\circ}$ and $A=A^\circ$ holds $\cP(\cF_\tau,P)$-q.s. Conversely, given $A^\circ\in \cFci_{\tau^\circ}$, we find
  $A\in\cF_\tau$ such that $A=A^\circ$ holds $\cP(\cF_\tau,P)$-q.s. We conclude that
  \begin{align*}
    E^P \big[P_1(\,\cdot\,|\cF_\tau)\1_\Lambda + P_2(\,\cdot\,|\cF_\tau)\1_{\Lambda^c}\big]
    = E^P \big[P_1(\,\cdot\,|\cFci_{\tau^\circ})\1_{\Lambda^\circ} + P_2(\,\cdot\,|\cFci_{\tau^\circ})\1_{(\Lambda^\circ)^c}\big].
  \end{align*}
  The right hand side is an element of $\cP$ by the stability under $\Fci$-pasting.
\end{proof}

We can now prove the optional sampling theorem for $\cE$-martingales; in particular, this establishes the $\F$-time-consistency of $\{\cE_t\}$ along general $\hF$-stopping times.

\begin{theorem}\label{th:DPPstop}
  Let $0\leq \sigma\leq \tau\leq T$ be stopping times, $X\in\cH$, and let
  $\cE(X)$ be the c\`adl\`ag $\cE$-martingale associated with $X$. Then
  \begin{equation}\label{eq:DPPStop}
    \cE_\sigma(X) = \mathop{\esssup^P}_{P'\in \cP(\hcF_\sigma,P)} E^{P'}[\cE_\tau(X)|\hcF_\sigma]\quad P\as\quad \mbox{for all } P\in\cP
  \end{equation}
  and in particular
  \begin{equation}\label{eq:DPPStopSpecial}
    \cE_\sigma(X) = \mathop{\esssup^P}_{P'\in \cP(\hcF_\sigma,P)} E^{P'}[X|\hcF_\sigma]\quad P\as\quad \mbox{for all } P\in\cP.
  \end{equation}
  Moreover, there exists for each $P\in\cP$ a sequence $P_n\in\cP(\hcF_\sigma,P)$ such that
  \begin{equation}\label{eq:DPPStopIncreasingSeq}
    \cE_\sigma(X) = \lim_{n\to\infty} E^{P_n}[X|\hcF_\sigma]\quad P\as
  \end{equation}
  with an increasing limit.
\end{theorem}

\begin{proof}
  Fix $P\in\cP$ and let $Y:=\cE(X)$.

  (i)~We first show the inequality ``$\geq$'' in~\eqref{eq:DPPStopSpecial}.
  By Proposition~\ref{pr:extension}(i), $Y$ is an $(\hF,P')$-supermartingale for all $P'\in\cP(\hcF_\sigma,P)$. Hence the (usual) optional sampling theorem implies the claim.

  (ii)~In the next two steps, we show the inequality ``$\leq$'' in~\eqref{eq:DPPStopSpecial}.
  In view of Lemma~\ref{le:stopFields}(ii) we may assume that $\sigma$ is an $\F^+$-stopping time,
  and then $(\sigma+1/n)\wedge T$ is an $\Fci$-stopping time for each $n\geq1$. For the time being, we also assume that
  $\sigma$ takes finitely many values. Let
  \[
     D_n:=\{k2^{-n}:\,k=0,1,\dots\}\cup\{T\}
  \]
  and define
  \[
    \sigma^n(\omega):=\inf \{t\in D_n:\,t\geq \sigma(\omega)+1/n\}\wedge T.
  \]
  Each $\sigma^n$ is an $\Fci$-stopping time taking finitely many values and $\sigma^n(\omega)$ decreases to $\sigma(\omega)$ for all $\omega\in\Omega$. Since the range of $\{\sigma, (\sigma^n)_n\}$ is countable, it follows from
  Proposition~\ref{pr:extension}(ii) that $\cE^\circ_{\sigma^n}(X)\to Y_\sigma$ $P$-a.s. Since $\|X\|_{L^1_\cP}<\infty$, the backward supermartingale convergence theorem \cite[Theorem~V.30]{DellacherieMeyer.82}
  implies that this convergence holds also in $L^1(P)$ and that
  \begin{equation}\label{eq:optionalSamplingProof}
    Y_\sigma=\lim_{n\to\infty} E^P[\cE^\circ_{\sigma^n}(X)|\hcF_\sigma]\quad P\as,
  \end{equation}
  where, by monotonicity, the $P$-a.s.\ convergence holds without passing to a subsequence.
  By Lemma~\ref{le:rawStopping} and Lemma~\ref{le:increasingSequence}, there exists for each $n$ a sequence $(P^n_k)_{k\geq1}$ in $\cP(\cFci_{\sigma^n},P)$ such that
  \[
    \cE^\circ_{\sigma^n}(X)
    =\mathop{\esssup^P}_{P'\in \cP(\cFci_{\sigma^n},P)} E^{P'}[X|\cFci_{\sigma^n}]
    =\lim_{k\to\infty} E^{P^n_k}[X|\cFci_{\sigma^n}]\quad P\as,
  \]
  where the limit is increasing. Moreover, using that
  \[
    \cF^+_{\sigma^{n+1}}=\big\{A\in\cFci_T:\, A \cap \{\sigma^{n+1}<t\}\in\cFci_t\mbox{ for }0\leq t\leq T\big\},
  \]
  the fact that $\sigma^n>\sigma^{n+1}$ on $\{\sigma^n<T\}$ is seen to imply that $\cF^+_{\sigma^{n+1}}\subseteq \cFci_{\sigma^n}$. Together
  with $\sigma\leq \sigma^{n+1}$ and Lemma~\ref{le:stopFields}(ii) we conclude that
  \begin{equation}\label{eq:optionalSamplingProof2}
    \hcF_\sigma \,\subseteq\, \hcF_{\sigma^{n+1}}\overset{\cP\mbox{\scriptsize{-q.s.}}}{=}  \cF^+_{\sigma^{n+1}}\,\subseteq \, \cFci_{\sigma^n} \quad\mbox{and hence}\quad \cP(\hcF_\sigma,P)\supseteq \cP(\cFci_{\sigma^n},P)
  \end{equation}
  for all $n$. Now monotone convergence yields
   \begin{align*}
    E^P[\cE^\circ_{\sigma^n}(X)|\hcF_\sigma]
    & = \lim_{k\to\infty} E^{P^n_k}[X|\hcF_\sigma]
    & \leq \mathop{\esssup^P}_{P'\in \cP(\hcF_\sigma,P)} E^{P'}[X|\hcF_\sigma]\quad P\as
  \end{align*}
  In view of~\eqref{eq:optionalSamplingProof}, this ends the proof of~\eqref{eq:DPPStopSpecial} for $\sigma$ taking finitely many values.

  (ii')~Now let $\sigma$ be general.  We approximate $\sigma$ by the decreasing sequence
  $\sigma^n:=\inf \{t\in D_n:\,t\geq \sigma\}\wedge T$ of stopping times with finitely many values. Then
  $\cE_{\sigma^n}(X)\equiv Y_{\sigma^n}\to Y_\sigma$ $P$-a.s.\ since $Y$ is c\`adl\`ag. The same arguments as for~\eqref{eq:optionalSamplingProof} show that
  \begin{equation}\label{eq:optionalSamplingProof3}
    Y_\sigma=\lim_{n\to\infty} E^P[\cE_{\sigma^n}(X)|\hcF_\sigma]\quad P\as
  \end{equation}
  By the two previous steps we have the representation~\eqref{eq:DPPStopSpecial} for $\sigma^n$.
  As in Lemma~\ref{le:increasingSequence}, it follows from the stability under $\F$-pasting (Lemma~\ref{le:pastingInF})
  that there exists for each $n$ a sequence $(P^n_k)_{k\geq1}$ in $\cP(\cF_{\sigma^n},P)\subseteq \cP(\cF_\sigma,P)$ such that
  \[
    \cE_{\sigma^n}(X)
    =\mathop{\esssup^P}_{P'\in \cP(\cF_{\sigma^n},P)} E^{P'}[X|\cF_{\sigma^n}]
    =\lim_{k\to\infty} E^{P^n_k}[X|\cF_{\sigma^n}]\quad P\as,
  \]
  where the limit is increasing and hence
   \begin{align*}
    E^P[\cE_{\sigma^n}(X)|\hcF_\sigma]
    & = \lim_{k\to\infty} E^{P^n_k}[X|\hcF_\sigma]
    & \leq \mathop{\esssup^P}_{P'\in \cP(\hcF_\sigma,P)} E^{P'}[X|\hcF_\sigma]\quad P\as
  \end{align*}
  Together with~\eqref{eq:optionalSamplingProof3}, this completes the proof of~\eqref{eq:DPPStopSpecial}.

  (iii)~We now prove~\eqref{eq:DPPStopIncreasingSeq}. Since $\sigma$ is general, the claim does not follow from the stability under pasting.  Instead, we use the construction of (ii'). Indeed, we have obtained
  $P^n_k\in\cP(\cF_{\sigma^n},P)$ such that
  \[
    Y_\sigma = \lim_{n\to\infty} \lim_{k\to\infty} E^{P^n_k}[X|\hcF_\sigma]\quad P\as
  \]
  Fix $n$. Since $\sigma^n$ is an $\F$-stopping time taking finitely many values and since
  $\hcF_{\sigma}\subseteq\hcF_{\sigma^n}$, it follows from the stability under $\F$-pasting (applied to $\sigma^n$) that the set $\{E^{P'}[X|\hcF_\sigma]:\, P'\in\cP(\hcF_{\sigma^n},P)\}$ is $P$-a.s.\ upward filtering,
  exactly as in the proof of Lemma~\ref{le:increasingSequence}.
  In view of $\cP(\cF_{\sigma^n},P)\subseteq \cP(\cF_{\sigma^{n+1}},P)$, it follows that for each $N\geq 1$ there exists $P^{(N)}\in\cP(\cF_{\sigma^N},P)$ such that
  \[
    E^{P^{(N)}}[X|\hcF_\sigma] = \max_{1\leq n\leq N} \max_{1\leq k\leq n} E^{P^n_k}[X|\hcF_\sigma]\quad P\as
  \]
  Since $\cP(\cF_{\sigma^N},P)\subseteq\cP(\hcF_\sigma,P)$, this yields the claim.

  (iv)~To prove~\eqref{eq:DPPStop}, we first express $\cE_\sigma(X)$ and $\cE_\tau(X)$ as essential suprema by using~\eqref{eq:DPPStopSpecial} both for $\sigma$ and for $\tau$. The inequality ``$\leq$'' is then immediate. The converse inequality follows by a monotone convergence argument exactly as in the proof of Proposition~\ref{pr:pastingAndTimeconsistency}(i), except that the increasing sequence is now obtained from~\eqref{eq:DPPStopIncreasingSeq} instead of Lemma~\ref{le:increasingSequence}.
\end{proof}

\subsection{Decomposition and 2BSDE for $\cE$-Martingales}

The next result contains the semimartingale decomposition of $\cE(X)$ under each $P\in\cP$ and can be seen as an analogue of the optional decomposition~\cite{ElKarouiQuenez.95} used in mathematical finance. In the context of $G$-expectations, such a result has also been referred to as ``$G$-martingale representation theorem''; see \cite{HuPeng.2010, SonerTouziZhang.2010rep, Song.10, XuZhang.09}. Those results are ultimately based on the PDE description of the $G$-expectation and are more precise than ours; in particular, they provide a single increasing process $K$ rather than a family $(K^P)_{P\in\cP}$ (but see Remark~\ref{rk:pathwInt}). On the other hand, we obtain an $L^1$-theory whereas those results require more integrability for $X$.

\begin{proposition}\label{pr:martingaleDecomp}
  Let $X\in\cH$. There exist
    \begin{enumerate}[topsep=3pt, partopsep=0pt, itemsep=1pt,parsep=2pt]
    \item an $\hF$-predictable process $Z^X$ with $\int_0^T |Z^X_s|^2\,d\br{B}_s<\infty$ $\cP$-q.s.,
    \item a family $(K^P)_{P\in\cP}$ of $\overline{\F}^P$-predictable processes such that all paths of $K^P$ are c\`adl\`ag nondecreasing and $E^P[|K^P_T|]<\infty$,
  \end{enumerate}
  such that
  \begin{equation}\label{eq:martDecompInProp}
    \cE_t(X)=\cE_0(X)+\sideset{^{(P)\hspace{-7pt}}}{}{\int_0^t} Z^X_s\,dB_s - K_t^P\quad\mbox{for all } 0\leq t\leq T,\quad P\as
  \end{equation}
  for all $P\in\cP$.
  The process $Z^X$ is unique up to $\{ds\times P,\,P\in\cP\}$-polar sets and $K^P$ is unique up to $P$-evanescence.
\end{proposition}

\begin{proof}
  We shall use arguments similar to the proof of~\cite[Theorem~4.5]{SonerTouziZhang.2010dual}.

  Let $P\in\cP$. It follows from Proposition~\ref{pr:extension}(i) that $Y:=\cE(X)$ is an $(\overline{\F}^P,P)$-supermartingale. We apply the Doob-Meyer decomposition in the filtered space $(\Omega,\overline{\F}^P, P)$ which  satisfies the usual conditions of right continuity and completeness. Thus we obtain an $(\overline{\F}^P, P)$-local martingale $M^P$ and an $\overline{\F}^P$-predictable increasing integrable process $K^P$, c\`adl\`ag and satisfying
  $M^P_0=K^P_0=0$, such that
  \[
    Y=Y_0+ M^P-K^P.
  \]
  By Lemma~\ref{le:MRPandVersions}, $(P,B)$ has the predictable representation property in $\overline{\F}^P$. Hence there exists an $\overline{\F}^P$-predictable process $Z^P$ such that
  \[
    Y=Y_0+ \sideset{^{(P)\hspace{-7pt}}}{}{\int} Z^P\,dB - K^P.
  \]

  The next step is to replace $Z^P$ by a process $Z^X$ independent of $P$.
  Recalling that $B$ is a continuous local martingale under each $P$, we have
  \begin{equation}\label{eq:constructionZ}
    \int Z^P\,d\br{B}^P = \br{Y,B}^P = BY - \sideset{^{(P)\hspace{-7pt}}}{}{\int} B\,dY - \sideset{^{(P)\hspace{-7pt}}}{}{\int} Y_-\,dB\quad P\as
  \end{equation}
  (Here and below, the statements should be read componentwise.)
  The last two integrals are It\^o integrals under $P$, but they can also be defined pathwise since the integrands are left limits of c\`adl\`ag processes which are bounded path-by-path.
  This is a classical construction from~\cite[Theorem~7.14]{Bichteler.81}; see also~\cite{Karandikar.95} for the same result in modern notation. To make explicit that the resulting process is $\F$-adapted, we recall the procedure for the example $\int Y_-\,dB$. One first defines for each $n\geq1$ the sequence of $\hF$-stopping times
  $\tau^n_0:=0$ and $\tau^n_{i+1}:=\inf\{t\geq \tau^n_i:\, |Y_t-Y_{\tau^n_i}|\geq 2^{-n}\}$. Then one defines
  $I^n$ by\vspace{-2pt}
  \[\vspace{-2pt}
    I^n_t:= Y_{\tau^n_k}(B_t-B_{\tau^n_k}) + \sum_{i=0}^{k-1} Y_{\tau^n_i}(B_{\tau^n_{i+1}}-B_{\tau^n_i})
    \quad\mbox{for}\quad\tau^n_k< t\leq\tau^n_{k+1},\quad k\geq0;
  \]
  clearly $I^n$ is again $\hF$-adapted and all its paths are c\`adl\`ag. Finally, we define
  \[
    I_t:=\limsup_{n\to\infty} I^n_t,\quad 0\leq t\leq T.
  \]
  Then $I$ is again $\hF$-adapted and it is a consequence of the Burkholder-Davis-Gundy inequalities that\vspace{-2pt}
  \[\vspace{-2pt}
    \sup_{0\leq t\leq T}\bigg| I^n_t-\sideset{^{(P)\hspace{-7pt}}}{}{\int_0^t}Y_-\,dB\bigg|\to 0\quad P\as
  \]
  for each $P$. Thus, outside a $\cP$-polar set, the limsup in the definition of $I$ exists as a limit uniformly in $t$ and $I$ has c\`adl\`ag paths. Since $\cP$-polar sets are contained in $\hcF_0$, we may redefine $I:=0$ on the exceptional set. Now $I$ is c\`adl\`ag $\hF$-adapted and coincides with the It\^o integral ${}^{(P)\hspace{-5pt}}\int Y_-\,dB$ up to $P$-evanescence, for all $P\in\cP$.

  We proceed similarly with the integral ${}^{(P)\hspace{-5pt}}\int B\,dY$ and obtain a definition for the right hand side of~\eqref{eq:constructionZ} which is
  $\hF$-adapted, continuous and independent of $P$. Thus we have defined $\br{Y,B}$ simultaneously for all $P\in\cP$, and we do the same for $\br{B}$. Let $\hat{a}=d\br{B}/dt$ be the (left) derivative in time of $\br{B}$, then $\hat{a}$ is $\Fci$-predictable and $\mathbb{S}^{>0}_d$-valued $P\times dt$-a.e.\ for all $P\in\cP$ by the definition of $\cP_S$. Finally, $Z^X:=\hat{a}^{-1}d\br{Y,B}/dt$ is an $\hF$-predictable process such that\vspace{-2pt}
  \[\vspace{-2pt}
    Y=Y_0+ \sideset{^{(P)\hspace{-7pt}}}{}{\int} Z^X\,dB -K^P\quad P\as\quad\mbox{for all }P\in\cP.
  \]
  We note that the integral is taken under $P$; see also Remark~\ref{rk:pathwInt} for a way to define it for all $P\in\cP$ simultaneously.
\end{proof}

The previous proof shows that a decomposition of the type~\eqref{eq:martDecompInProp} exists for all c\`adl\`ag $(\F,\cP)$-supermartingales, and not just for $\cE$-martingales.
As a special case of Proposition~\ref{pr:martingaleDecomp}, we obtain a representation for symmetric $\cE$-martingales. The following can be seen as a generalization of the corresponding results for $G$-expectations given in~\cite{SonerTouziZhang.2010rep, Song.10, XuZhang.09}.

\begin{corollary}\label{co:symmMartDecomp}
  Let $X\in\cH$ be such that $-X\in\cH$. The following are equivalent:
    \begin{enumerate}[topsep=3pt, partopsep=0pt, itemsep=1pt,parsep=2pt]
    \item $\cE(X)$ is a symmetric $\cE$-martingale; i.e., $\cE(-X)=-\cE(X)$ $\cP$-q.s.
    \item There exists an $\hF$-predictable process $Z^X$ with $\int_0^T |Z^X_s|^2\,d\br{B}_s<\infty$ $\cP$-q.s.\ such that
    \vspace{-5pt}
    \[\vspace{-3pt}
      \cE_t(X)=\cE_0(X)+ \int_0^t Z^X_s\,dB_s \quad\mbox{for all } 0\leq t\leq T,\quad \cP\qs,
    \]
    where the integral can be defined universally for all $P$ and $\int Z^X\,dB$ is an $(\F,P)$-martingale for all $P\in\cP$.
  \end{enumerate}
  In particular, any symmetric $\cE$-martingale has continuous trajectories $\cP$-q.s.
\end{corollary}

\begin{proof}
  The implication (ii)$\Rightarrow$(i) is clear from Proposition~\ref{pr:extension}(iii). Conversely, given~(i),
  Proposition~\ref{pr:extension}(i) yields that both $\cE(X)$ and $-\cE(X)$ are $\cP$-supermartingales, hence $\cE(X)$ is a (true) $\cP$-martingale. It follows that the increasing processes $K^P$ have to satisfy $K^P\equiv0$ and~\eqref{eq:martDecompInProp} becomes $\cE(X)=\cE_0(X)+{}^{(P)\hspace{-5pt}}\int Z^X\,dB$. In particular, the stochastic integral can be defined universally by setting $\int Z^X\,dB:=\cE(X)-\cE_0(X)$.
\end{proof}

\begin{remark}
  (a)~Without the martingale condition in Corollary~\ref{co:symmMartDecomp}(ii), the implication (ii)$\Rightarrow$(i) would fail even for $\cP=\{P_0\}$, in which case Corollary~\ref{co:symmMartDecomp} is simply the Brownian martingale representation theorem.

  (b)~Even if it is symmetric, $\cE(X)$ need not be a $\cP$-modification of the family $\{\cE^\circ_t(X),\, t\in[0,T]\}$; in fact, the $\cE$-martingale in Example~\ref{ex:counterexModification} is symmetric.
  However, the situation changes if the symmetry assumption is imposed directly on $\{\cE^\circ_t(X)\}$. We call $\{\cE^\circ_t(X)\}$ symmetric if $\cE^\circ_t(-X)=-\cE^\circ_t(X)$ $\cP$-q.s.\ for all $t\in[0,T]$.
  \begin{itemize}
    \item\emph{If $\{\cE_t^\circ(X)\}$ symmetric, then
    $\cE(X)$ is a symmetric $\cE$-martingale and a $\cP$-modification of $\{\cE^\circ_t(X)\}$.}
  \end{itemize}
  Indeed, the assumption implies that $\{\cE^\circ_t(X)\}$ is an $(\Fci,P)$-martingale for each $P\in\cP$ and so the process $\cE(X)$ of right limits (cf.~Proposition~\ref{pr:extension}(ii)) is the usual c\`adl\`ag $P$-modification of $\{\cE^\circ_t(X)\}$, for all $P$.
\end{remark}

Next, we represent the pair $(\cE(X),Z^X)$ from Proposition~\ref{pr:martingaleDecomp} as the solution of a 2BSDE. The following definition is essentially from~\cite{SonerTouziZhang.2010bsde}.

\begin{definition}\label{def:2BSDE}
  Let $X\in L^1_\cP$ and consider a pair $(Y,Z)$ of processes with values in $\R\times\R^d$ such that
  $Y$ is c\`adl\`ag $\hF$-adapted while
  $Z$ is $\hF$-predictable and $\int_0^T |Z_s|^2\,d\br{B}_s<\infty$ $\cP$-q.s.
  Then $(Y,Z)$ is called a \emph{solution} of the 2BSDE~\eqref{eq:2bsde} if there exists a family $(K^P)_{P\in\cP}$ of $\overline{\F}^P$-adapted increasing processes satisfying $E^P[|K^P_T|]<\infty$ such that
  \begin{equation}\label{eq:2bsde}
   Y_t = X - \sideset{^{(P)\hspace{-7pt}}}{}{\int_t^T} Z_s\,dB_s + K_T^P-K_t^P,\quad 0\leq t\leq T,\quad P\as\quad\mbox{for all }P\in\cP
  \end{equation}
  and such that the following minimality condition holds for all $0\leq t\leq T$:
  \begin{equation}\label{eq:minimal}
    \mathop{\essinf^P}_{P'\in \cP(\hcF_t,P)} E^{P'}\big[K_T^{P'}-K_t^{P'}\big|\hcF_t\big]=0\quad P\as \quad\mbox{for all }P\in\cP.
  \end{equation}
\end{definition}

We note that~\eqref{eq:minimal} is essentially the $\cE$-martingale condition~\eqref{eq:Yesssup}: if the processes $K^P$ can be aggregated into a single process $K$ and $K_T\in\cH$, then $-K=\cE(-K_T)$.
Regarding the aggregation of $(K^P)$, see also Remark~\ref{rk:pathwInt}.

A second notion is needed to state the main result. A c\`adl\`ag process $Y$ is said to be \emph{of class (D,$\cP$)} if
the family $\{Y_\sigma\}_\sigma$ is uniformly integrable under $P$ for all $P\in\cP$, where $\sigma$ runs through all
$\hF$-stopping times. As an example, we have seen in Corollary~\ref{co:symmMartDecomp} that all symmetric $\cE$-martingales are
of class~(D,$\cP$). (Of course, it is important here that we work with a finite time horizon $T$.)
For $p\in[1,\infty)$, we define $\|X\|_{L^p_\cP}=:\sup_{P\in\cP} E[|X|^p]^{1/p}$ as well as $\cH^p:=\{X\in\cH:\, |X|^p\in\cH\}$.

\begin{lemma}\label{le:classD}
  If $X\in\cH^p$ for some $p\in (1,\infty)$, then $\cE(X)$ is of class (D,$\cP$).
\end{lemma}

\begin{proof}
  Let $P\in\cP$. If $\sigma$ is an $\hF$-stopping time, Jensen's inequality and~\eqref{eq:DPPStopSpecial} yield that
  \[
    |\cE_\sigma(X)|^p \leq  \mathop{\esssup^P}_{P'\in \cP(\hcF_\sigma,P)} E^{P'}[|X|^p|\hcF_\sigma]=\cE_\sigma(|X|^p)\quad P\as
  \]
  In particular, $\|\cE_\sigma(X)\|^p_{L^p(P)}\leq E^P[\cE_\sigma(|X|^p)]$ and thus Lemma~\ref{le:MRPandVersions} yields
  \[
    \|\cE_\sigma(X)\|^p_{L^p(P)} \leq E^P[\cE_\sigma(|X|^p)|\cF_0] \leq \mathop{\esssup^P}_{P'\in \cP(\hcF_0,P)} E^{P'}[\cE_\sigma(|X|^p)|\hcF_0] \quad P\as
  \]
  The right hand side $P$-a.s.\ equals $\cE_0(|X|^p)$ by~\eqref{eq:DPPStop}, so we conclude with~\eqref{eq:ineqModification} that
  \[
    \|\cE_\sigma(X)\|^p_{L^p(P)}
    \leq  \cE_0(|X|^p)
    \leq \sup_{P'\in\cP} E^{P'}[|X|^p]
    = \|X\|^p_{L^p_\cP}<\infty \quad P\as
  \]
  Therefore, the family $\{\cE_\sigma(X)\}_\sigma$ is bounded in $L^p(P)$ and in particular uniformly integrable under $P$. This holds for all $P\in\cP$.
\end{proof}

We can now state the main result of this section.

\pagebreak[2]

\begin{theorem}\label{th:2bsde}
  Let $X\in\cH$.
  \begin{enumerate}[topsep=3pt, partopsep=0pt, itemsep=1pt,parsep=2pt]
    \item The pair $(\cE(X),Z^X)$ is the minimal solution of the 2BSDE~\eqref{eq:2bsde}; i.e., if $(Y,Z)$ is another solution, then
            $\cE(X)\leq Y$ $\cP$-q.s.
    \item If $(Y,Z)$ is a solution of~\eqref{eq:2bsde} such that $Y$ is of class (D,$\cP$), then
        $(Y,Z)=(\cE(X),Z^X)$.
  \end{enumerate}
  In particular, if $X\in\cH^p$ for some $p>1$, then $(\cE(X), Z^X)$ is the unique solution of~\eqref{eq:2bsde} in the class (D,$\cP$).
\end{theorem}

\begin{proof}
  (i)~Let $P\in\cP$. To show that $(\cE(X),Z^X)$ is a solution, we only have to show that $K^P$ from the decomposition~\eqref{eq:martDecompInProp} satisfies the minimality condition~\eqref{eq:minimal}.
  We denote this decomposition by $\cE(X)=\cE_0(X)+M^P-K^P$.
  It follows from Proposition~\ref{pr:extension}(i) that $\cE(X)$ is an $(\overline{\F}^P,P)$-supermartingale.
  As $K^P\geq0$, we deduce that
  \[
    \cE_0(X)+ M^P \geq \cE(X) \geq E^{P}[X|\overline{\F}^P]\quad P\as,
  \]
  where $E^{P}[X|\overline{\F}^P]$ denotes the c\`adl\`ag $(\overline{\F}^P,P)$-martingale with terminal value $X$. Hence
  $M^P$ is a local $P$-martingale bounded from below by a $P$-martingale and thus $M^P$ is an $(\hF,P)$-supermartingale
  by a standard argument using Fatou's lemma. This holds for all $P\in\cP$. Therefore, \eqref{eq:Yesssup} yields
  \begin{align*}
  0 &= \cE_t(X)-\mathop{\esssup^P}_{P'\in \cP(\hcF_t,P)} E^{P'}[X|\hcF_t]\\
  &= \mathop{\essinf^P}_{P'\in \cP(\hcF_t,P)} E^{P'}\big[\cE_t(X)-\cE_T(X)\big|\hcF_t\big]\\
  &= \mathop{\essinf^P}_{P'\in \cP(\hcF_t,P)} E^{P'}\big[M_t^{P'}-M_T^{P'} + K_T^{P'}-K_t^{P'}\big|\hcF_t\big]\\
  &\geq \mathop{\essinf^P}_{P'\in \cP(\hcF_t,P)} E^{P'}\big[K_T^{P'}-K_t^{P'}\big|\hcF_t\big]\quad P\as\quad\mbox{for all }P\in\cP.
  \end{align*}
  Since $K^{P'}$ is nondecreasing, the last expression is also nonnegative and~\eqref{eq:minimal} follows. Thus $(\cE(X),Z^X)$ is a solution.

  To prove the minimality, let $(Y,Z)$ be another solution of~\eqref{eq:2bsde}.
  It follows from~\eqref{eq:2bsde} that $Y$ is a local $(\hF,P)$-supermartingale for all $P\in\cP$.
  As above, the integrability of $X$ implies that $Y_0+{}^{(P)\hspace{-5pt}}\int Z\,dB$ is bounded below by a $P$-martingale. Noting also that $Y_0$ is $P$-a.s.\ equal to a constant by Lemma~\ref{le:MRPandVersions}, we deduce that
  ${}^{(P)\hspace{-5pt}}\int Z\,dB$ and $Y$ are $(\hF,P)$-supermartingales. Since $Y$ is c\`adl\`ag and $Y_T=X$, the minimality property in Proposition~\ref{pr:extension}(i) shows that $Y\geq \cE(X)$ $\cP$-q.s.

  (ii)~If in addition $Y$ is of class (D,$\cP$), then ${}^{(P)\hspace{-5pt}}\int Z\,dB$ is a true $P$-martingale by the Doob-Meyer theorem and we have
  \begin{align*}
  0 & = \mathop{\essinf^P}_{P'\in \cP(\hcF_t,P)} E^{P'}\big[K_T^{P'}-K_t^{P'}\big|\hcF_t\big]\\
    & = Y_t-\mathop{\esssup^P}_{P'\in \cP(\hcF_t,P)} E^{P'}[X|\hcF_t] \\
    & = Y_t -\cE_t(X)\quad P\as\quad\mbox{for all }P\in\cP.
  \end{align*}
   The last statement in the theorem follows from Lemma~\ref{le:classD}.
\end{proof}

\begin{remark}\label{rk:pathwInt}
  If we use axioms of set theory stronger than the usual ZFC, such as the Continuum Hypothesis, then the integrals $\{{}^{(P)\hspace{-5pt}}\int Z\,dB\}_{P\in\cP}$ can be aggregated into a single (universally measurable) continuous process, denoted by $\int Z\,dB$,  for any $Z$ which is $B$-integrable under all $P\in\cP$. This follows from a recent result on pathwise stochastic integration, cf.~\cite{Nutz.11int}.
  In Proposition~\ref{pr:martingaleDecomp}, we can then aggregate the family $(K^P)_{P\in\cP}$ of increasing processes into a single process $K$ by setting
  $K:=\cE_0(X) - \cE(X) + \int Z^X\,dB$. Moreover, we can strengthen Theorem~\ref{th:2bsde} by asking for a universal process $K$ the Definition~\ref{def:2BSDE} of the 2BSDE.
\end{remark}

\subsection{Application to Superhedging and Replication}

We now turn to the interpretation of the previous results for the superhedging problem.
Let $H$ be an $\R^d$-valued $\F$-predictable process satisfying $\int_0^T |H_s|^2\,d\br{B}_s<\infty$ $\cP$-q.s. Then $H$ is called an \emph{admissible} trading strategy if ${}^{(P)\hspace{-5pt}}\int H\,dB$ is a $P$-supermartingale for all $P\in\cP$. (We do not insist that the integral be defined without reference to $P$, since this is not necessary economically.
But see also Remark~\ref{rk:pathwInt}.) As usual in continuous-time finance, this definition excludes ``doubling strategies''. We have seen in the proof of Theorem~\ref{th:2bsde} that $Z^X$ is admissible for $X\in\cH$. The minimality property in Proposition~\ref{pr:extension}(i) and  the existence of the decomposition~\eqref{eq:martDecompInProp} yield the following conclusion: $\cE_0(X)$ is the minimal $\cF_0$-measurable initial capital which allows to superhedge $X$; i.e., $\cE_0(X)$ is the $\cP$-q.s.\ minimal $\cF_0$-measurable random variable $\xi_0$ such that there exists an admissible strategy $H$ satisfying
\[
  \xi_0+ \sideset{^{(P)\hspace{-7pt}}}{}{\int_0^T} H_s\,dB_s \geq X\quad P\as\quad\mbox{for all }P\in\cP.
\]
Moreover, the ``overshoot'' $K^P$ for the strategy $Z^X$ satisfies the minimality condition~\eqref{eq:minimal}.

As seen in Example~\ref{ex:counterexModification}, the $\cF_0$-superhedging price $\cE_0(X)$ need not be a constant, and therefore it is debatable whether it is a good choice for a conservative price, in particular if the raw filtration $\Fci$ is seen as the initial information structure for the model. Indeed, the following illustration shows that knowledge of $\cF_0$ can be quite significant. Consider a collection $(a_i)$ of positive constants and
$\cP=\{P^\alpha:\alpha\equiv a_i \mbox{ for some }i\}$. (Such a set $\cP$ can indeed satisfy the assumptions of this section.)
In this model, knowledge of $\cF_0$ completely removes the volatility uncertainty since $\cF_0$ contains the sets
\[
  A_i:=\Big\{\limsup_{t\to0} t^{-1}\br{B}_t = \liminf_{t\to 0} t^{-1}\br{B}_t =a_i \Big\}\in\cFci_{0+}
\]
which form a $\cP$-q.s.\ partition of $\Omega$.
Hence, one may want to use the more conservative choice
\[
  x=\cE^\circ_0(X)=\sup_{P\in\cP} E^P[X]=\inf\{y\in\R:\, y\geq \cE_0(X)\}
\]
as the price. This value can be embedded into the $\cE$-martingale as follows.
Let $\cF_{0-}$ be the smallest $\sigma$-field containing the $\cP$-polar sets, then $\cF_{0-}$ is trivial $\cP$-q.s.
If we adjoin $\cF_{0-}$ as a new initial state to the filtration $\F$, we can extend $\cE(X)$ by setting
\[
  \cE_{0-}(X):=\sup_{P\in\cP} E^P[X],\quad X\in\cH.
\]
The resulting process $\{\cE_t(X)\}_{t\in [-0,T]}$ satisfies the properties from Proposition~\ref{pr:extension} in the extended filtration and in particular the constant $x=\cE_{0-}(X)$ is the $\cF_{0-}$-superhedging price of $X$.
(Of course, all this becomes superfluous in the case where $\cE(X)$ is a $\cP$-modification of
$\{\cE^\circ_t(X)\}$.)

In the remainder of the section, we discuss replicable claims and adopt the previously mentioned conservative choice.%

\begin{definition}
  A random variable $X\in\cH$ is called \emph{replicable} if there exist a constant $x\in\R$ and an
  $\hF$-predictable process $H$ with $\int_0^T |H_s|^2\,d\br{B}_s<\infty$ $\cP$-q.s.\ such that
  \begin{equation}\label{eq:replicable}
    X=x+\sideset{^{(P)\hspace{-7pt}}}{}{\int_0^T} H_t\,dB_t\quad P\as\quad\mbox{for all }P\in\cP
  \end{equation}
  and such that ${}^{(P)\hspace{-5pt}}\int H\,dB$ is an $(\hF,P)$-martingale for all $P\in\cP$.
\end{definition}

The martingale assumption is needed to avoid strategies which ``throw away'' money. Moreover, as in Corollary~\ref{co:symmMartDecomp}, the stochastic integral can necessarily be defined without reference to $P$, by setting
\mbox{$\int H\,dB:=\cE(X)-x$}. The following result is an analogue of the standard characterization of replicable claims in incomplete markets (e.g.,~\cite[p.\,182]{DelbaenSchachermayer.06}).

\begin{proposition}
  Let $X\in\cH$ be such that $-X\in\cH$. The following are equivalent:
  \begin{enumerate}[topsep=3pt, partopsep=0pt, itemsep=1pt,parsep=2pt]
    \item $\cE(X)$ is a symmetric $\cE$-martingale and $\cE_0(X)$ is constant $\cP$-q.s.
    \item $X$ is replicable.
    \item There exists $x\in\R$ such that $E^P[X]=x$ for all $P\in\cP$.
  \end{enumerate}
\end{proposition}

\begin{proof}
  The equivalence (i)$\Leftrightarrow$(ii) is immediate from Corollary~\ref{co:symmMartDecomp} and
  the implication (ii)$\Rightarrow$(iii) follows by taking expectations in~\eqref{eq:replicable}. Hence we prove (iii)$\Rightarrow$(ii).
  By~\eqref{eq:ineqModification} we have $\cE_0(-X)\leq \sup_{P\in\cP} E^P[-X]=-x$ and similarly $\cE(X)\leq x$. Thus, given $P\in\cP$, the decompositions~\eqref{eq:martDecompInProp} of $\cE(-X)$ and $\cE(X)$ show that
  \begin{equation}\label{eq:replicationProofIneqs}
    -X \leq -x + \sideset{^{(P)\hspace{-7pt}}}{}{\int_0^T} Z^{-X}\,dB \quad\mbox{and}\quad X \leq x + \sideset{^{(P)\hspace{-7pt}}}{}{\int_0^T} Z^X\,dB \quad P\as
  \end{equation}
  Adding the inequalities yields $0\leq {}^{(P)\hspace{-5pt}}\int_0^T (Z^{-X} +Z^X) \,dB$ $P$-a.s.
  As we know from the proof of Theorem~\ref{th:2bsde} that the integrals of
  $Z^X$ and $Z^{-X}$ are supermartingales, it follows that
  ${}^{(P)\hspace{-5pt}}\int_0^T Z^{-X}\,dB = - {}^{(P)\hspace{-5pt}}\int_0^T Z^X\,dB$ $P$-a.s. Now~\eqref{eq:replicationProofIneqs} yields that $X=x+ {}^{(P)\hspace{-5pt}}\int_0^T Z^X\,dB$.
  In view of~(iii), this integral is a supermartingale with constant expectation, hence a martingale.
\end{proof}

\section{Uniqueness of Time-Consistent Extensions}\label{se:uniqueness}

In the introduction, we have claimed that $\{\cE^\circ_t(X)\}$ as in~\eqref{eq:aggreg} is the natural dynamic extension
of the static sublinear expectation $X\mapsto \sup_{P\in\cP}E^P[X]$. In this section, we add some substance to this claim by showing that the extension is unique under suitable assumptions. (We note that by Proposition~\ref{pr:pastingAndTimeconsistency}, the question of existence is essentially reduced to the technical problem of aggregation.)

The setup is as follows. We fix a nonempty set $\cP$ of probability measures on $(\Omega,\cFci_T)$; it is not important whether $\cP$ consists of martingale laws. On the other hand, we impose additional structure on the set of random variables. In this section, we consider a chain of vector spaces $(\cH_t)_{0\leq t\leq T}$ satisfying
\[
  \R=\cH_0\subseteq \cH_s\subseteq \cH_t\subseteq \cH_T=:\cH \subseteq L^1_\cP,\quad 0\leq s\leq t\leq T.
\]
We assume that $X,Y\in\cH_t$ implies $X\wedge Y, X\vee Y\in\cH_t$, and $XY\in \cH_t$ if in addition $Y$ is bounded.
As before, $\cH$ should be seen as the set of financial claims. The elements of $\cH_t$ will serve as ``test functions''; the main example to have in mind is $\cH_t=\cH\cap L^1_\cP(\cFci_t)$.
We consider a family $(\E_t)_{0\leq t\leq T}$ of mappings
\[
  \E_t: \cH \to L^1_\cP(\cFci_t)
\]
and think of $(\E_t)$ as a dynamic extension of $\E_0$. Our aim is to find conditions under which $\E_0$ already determines the whole family
$(\E_t)$, or more precisely, determines $\E_t(X)$ up to a $\cP$-polar set for all $X\in\cH$ and $0\leq t\leq T$.

\begin{definition}\label{def:poshom}
  The family $(\E_t)_{0\leq t\leq T}$ is called
  \emph{$(\cH_t)$-positively homogeneous} if for all $t\in[0,T]$ and $X\in \cH$,
  \[
    \E_t(X\varphi)=\E_t(X)\varphi\quad\cP\qs\quad\mbox{for all bounded nonnegative }\varphi\in\cH_t.
  \]
\end{definition}
Note that this property excludes trivial extensions of $\E_0$. Indeed, given $\E_0$, we can always define the (time-consistent) extension
\[
  \E_t(X):=
  \begin{cases}
    \E_0(X), & 0\leq t<T, \\
    X, & t=T,
  \end{cases}
\]
but this family $(\E_t)$ is not $(\cH_t)$-positively homogeneous for nondegenerate choices of $(\cH_t)$.

To motivate the next definition, we first recall that in the classical setup under a reference measure $P_*$, strict monotonicity of $\E_0$ is the crucial condition for uniqueness of extensions; i.e., $X\geq Y$ $P_*$-a.s.\ and $P_*\{X>Y\}>0$ should imply that $\E_0(X)>\E_0(Y)$.
In our setup with singular measures, the corresponding condition is too strong. E.g., for
$\E_0(\cdot)=\sup_{P\in\cP}E^P[\,\cdot\,]$, it is completely reasonable to have random variables $X\geq Y$ satisfying $\E_0(X)=\E_0(Y)$ and $P_1\{X>Y\}>0$ for some $P_1\in\cP$, since the suprema can be attained at some $P_2\in\cP$ whose support is disjoint from $\{X>Y\}$. In the following definition, we allow for an additional localization by a test function.

\begin{definition}\label{def:localStrictMon}
  We say that $\E_0$ is \emph{$(\cH_t)$-locally strictly monotone} if for every
  $t\in[0,T]$ and any $X,Y\in\cH_t$ satisfying $X\geq Y$ $\cP$-q.s.\ and $P(X>Y)>0$ for some $P\in\cP$, there exists $f\in\cH_t$ such that $0\leq f\leq 1$ and
  \[
    \E_0(Xf)>\E_0(Yf).
  \]
\end{definition}

Here the delicate point is the regularity required for $f$. Indeed, one is tempted to try
$f:=\1_{\{X>Y+\delta\}}$ (for some constant $\delta>0$), but in applications the definition of $\cH_t$ may exclude this choice and require a more refined construction. We defer this task to Proposition~\ref{pr:ExamplesLocalStrictMon} and first show how local strict monotonicity yields uniqueness.

\begin{proposition}\label{pr:uniqueness}
  Let $\E_0$ be $(\cH_t)$-locally strictly monotone. Then there exists at most one extension of $\E_0$ to a
  family $(\E_t)_{0\leq t\leq T}$  which is $(\cH_t)$-positively homogeneous  and satisfies $\E_t(\cH)\subseteq \cH_t$ and $\E_0\circ \E_t=\E_0$ on $\cH$.
\end{proposition}

\begin{proof}
  Let $(\E_t)$ and $(\widetilde{\E}_t)$ be two such extensions and suppose for contradiction that $\E_t(X)\neq\tbE_t(X)$
  for some $X\in\cH$; i.e., there exists $P\in\cP$ such that either $P\{\E_t(X)>\tbE_t(X)\}>0$ or $P\{\E_t(X)<\tbE_t(X)\}>0$. Without loss of generality, we focus on the first case. Define
  \[
    \varphi:=\big(\big[\E_t(X)-\tbE_t(X)\big]\vee 0\big) \wedge 1.
  \]
  Then $\varphi\in\cH_t$, since $\cH_t$ is a lattice containing the constant functions; moreover, $0\leq \varphi\leq 1$ and
  $\{\varphi=0\}=\{\E_t(X)\leq \tbE_t(X)\}$. Setting $X':=X\varphi$ and using the positive homogeneity, we arrive at
  \[
    \E_t(X')\geq \widetilde{\E}_t(X')\quad\mbox{and}\quad P\big\{\E_t(X')> \widetilde{\E}_t(X')\big\}>0.
  \]
  By local strict monotonicity there exists $f\in\cH_t$ such that $0\leq f\leq1$ and
  $\E_0\big(\E_t(X')f\big)> \E_0\big(\widetilde{\E}_t(X')f\big)$. Now $\E_0=\E_0\circ \E_t$ yields that
  \[
    \E_0(X'f)
    =\E_0\big(\E_t(X')f\big)> \E_0\big(\widetilde{\E}_t(X')f\big)
    =\widetilde{\E}_0(X'f),
  \]
  which contradicts $\E_0=\widetilde{\E}_0$.
\end{proof}

We can extend the previous result by applying it on dense subspaces. This relaxes the assumption that
$\E_t(\cH)\subseteq \cH_t$ and simplifies the verification of local strict monotonicity since one can choose convenient spaces of test functions. Consider a chain of spaces $(\hat{\cH}_t)_{0\leq t\leq T}$ satisfying the same assumptions as $(\cH_t)_{0\leq t\leq T}$ and such that $\hat{\cH}_T$ is a $\|\cdot\|_{L^1_\cP}$-dense subspace of $\cH$. We say that
$(\E_t)_{0\leq t\leq T}$ is \emph{$L^1_\cP$-continuous} if
\[
  \E_t: \big(\cH,\|\cdot\|_{L^1_\cP}\big)\to \big(L^1_\cP(\cFci_t),\|\cdot\|_{L^1_\cP}\big)
\]
is continuous for every $t$. We remark that the motivating example $(\cE^\circ_t)$ from Assumption~\ref{as:pastingAndAggreg} satisfies this property (it is even Lipschitz continuous).

\begin{corollary}\label{co:denseSubspace}
  Let $\E_0$ be $(\hat{\cH}_t)$-locally strictly monotone. Then there exists at most one extension of $\E_0$ to an
  $L^1_\cP$-continuous family $(\E_t)_{0\leq t\leq T}$ on $\cH$ which is $(\hat{\cH}_t)$-positively homogeneous and satisfies $\E_t(\hat{\cH}_T)\subseteq \hat{\cH}_t$ and \mbox{$\E_0\circ \E_t=\E_0$} on $\hat{\cH}_T$.
\end{corollary}

\begin{proof}
  Proposition~\ref{pr:uniqueness} shows that $\E_t(X)$ is uniquely determined for $X\in\hat{\cH}_T$. Since $\hat{\cH}_T\subseteq \cH$ is dense and $\E_t$ is continuous, $\E_t$ is also determined on $\cH$.
\end{proof}

In our last result, we show that  $\E_0(\cdot)=\sup_{P\in\cP}E^P[\,\cdot\,]$ is $(\cH_t)$-locally strictly monotone in certain cases.
The idea here is that we already have an extension $(\E_t)$ (as in Assumption~\ref{as:pastingAndAggreg}), whose uniqueness we try to establish.
We denote by $C_b(\Omega)$ the set of bounded continuous functions on $\Omega$ and by $C_b(\Omega_t)$ the $\cFci_t$-measurable functions in $C_b(\Omega)$, or equivalently the bounded functions which are continuous with respect to $\|\omega\|_t:=\sup_{0\leq s\leq t} |\omega_s|$. Similarly, $\UC_b(\Omega)$ and $\UC_b(\Omega_t)$ denote the sets of bounded uniformly continuous functions. We also define $\mathbb{L}^1_{c,\cP}$ to be the closure of $C_b(\Omega)$ in $L^1_\cP$, while
$\mathbb{L}^\infty_{c,\cP}$ denotes the $\cP$-q.s.\ bounded elements of $\mathbb{L}^1_{c,\cP}$. Finally, $\mathbb{L}^\infty_{c,\cP}(\cFci_t)$ is obtained similarly from $C_b(\Omega_t)$, while $\mathbb{L}^\infty_{uc,\cP}(\cFci_t)$ is the space obtained when starting from $\UC_b(\Omega_t)$ instead of $C_b(\Omega_t)$.

\begin{proposition}\label{pr:ExamplesLocalStrictMon}
  Let $\E_0(\cdot)=\sup_{P\in\cP}E^P[\,\cdot\,]$. Then $\E_0$ is $(\cH_t)$-locally strictly monotone for each of the cases
  \begin{enumerate}[topsep=3pt, partopsep=0pt, itemsep=1pt,parsep=2pt]
    \item $\cH_t= C_b(\Omega_t)$,
    \item $\cH_t= \UC_b(\Omega_t)$,
    \item $\cH_t=\mathbb{L}^\infty_{c,\cP}(\cFci_t)$,
    \item $\cH_t=\mathbb{L}^\infty_{uc,\cP}(\cFci_t)$.
  \end{enumerate}
\end{proposition}

Together with Corollary~\ref{co:denseSubspace}, this yields a uniqueness result for extensions.
Before giving the proof, we indicate some examples covered by this result; see also Example~\ref{ex:Gexp}.
The domain of $(\E_t)$ is $\cH=\mathbb{L}^1_{uc,\cP}$ in both cases. (This statement implicitly uses the fact that $\mathbb{L}^1_{uc,\cP}=\mathbb{L}^1_{c,\cP}$ when $\cP$ is tight; cf.\ the proof of~\cite[Proposition~5.2]{Nutz.10Gexp}.)

(a) Let $(\E_t)$ be the $G$-expectation as introduced in~\cite{Peng.07, Peng.08}. Then
Corollary~\ref{co:denseSubspace} applies: if $\hat{\cH}_t$ is any of the spaces in (i)--(iv), the invariance property $\E_t(\hat{\cH}_T)\subseteq \hat{\cH}_t$ is satisfied and $\hat{\cH}_T$ is dense in $\cH$.

(b) Using the construction given in~\cite{Nutz.10Gexp}, the $G$-expectation can be extended to the case when there is no finite upper bound for the volatility. This corresponds to a possibly infinite  function $G$ (and then $\cP$ need not be tight). Here Corollary~\ref{co:denseSubspace} applies with $\hat{\cH}_t=\UC_b(\Omega_t)$ since $\E_t(\hat{\cH}_T)\subseteq \hat{\cH}_t$ is satisfied by the remark stated after \cite[Corollary~3.6]{Nutz.10Gexp}, or also with $\cH_t=\mathbb{L}^\infty_{uc,\cP}(\cFci_t)$.

\begin{proof}[Proof of Proposition~\ref{pr:ExamplesLocalStrictMon}]
  Fix $t\in[0,T]$. All topological notions in this proof are expressed with respect to $d(\omega,\omega'):=\|\omega-\omega'\|_t$. Let $X,Y\in\cH_t$ be such that $X\geq Y$ $\cP$-q.s.\ and $P_*(X>Y)>0$ for some $P_*\in\cP$. By translating and multiplying with positive constants, we may assume that $1\geq X\geq Y\geq 0$. We prove the cases (i)--(iv) separately.

  (i)~Choose $\delta>0$ small enough so that $P_*\{X\geq Y+2\delta\}>0$ and let
  \[
    A_1:=\{X\geq Y + 2\delta\},\quad A_2:=\{X\leq Y+\delta\}.
  \]
  Then $A_1$ and $A_2$ are disjoint closed sets and
  \begin{equation}\label{eq:fDefProofStrictMon}
    f(\omega):=\frac{d(\omega,A_2)}{d(\omega,A_1)+d(\omega,A_2)}
  \end{equation}
  is a continuous function satisfying $0\leq f\leq 1$ as well as $f=0$ on $A_2$ and $f=1$ on $A_1$.
  It remains to check that
  \[
    \E_0(Xf)> \E_0(Yf), \quad\mbox{i.e.,}\quad \sup_{P\in\cP} E^P[Xf]>\sup_{P\in\cP} E^P[Yf].
  \]
  If $\E_0(Yf)=0$, the observation that
  $\E_0(Xf)\geq E^{P_*}[Xf]\geq 2\delta P_*(A_1)>0$
  already yields the proof.

  Hence, we may assume that $\E_0(Yf)>0$.
  For $\eps>0$, let $P_\eps\in\cP$ be such that $E^{P_\eps}[Yf]\geq \E_0(Yf)-\eps$. Since $X> Y+\delta$ on $\{f>0\}$ and since $0\leq Y\leq 1$, we have
  $
    Xf\geq (Y+\delta)f \geq (Y+\delta Y)f
  $
  and therefore
  \begin{align*}
    \E_0(Xf)
    &\geq \limsup_{\eps\to 0} E^{P_\eps}[(Y+\delta Y)f]\\
    & = \limsup_{\eps\to 0} \, (1+\delta) E^{P_\eps}[Yf]\\
    &= (1+\delta) \, \E_0(Yf).
  \end{align*}
  As $\delta>0$ and $\E_0(Yf)>0$, this ends the proof of~(i).

  (ii) The proof for this case is the same; we merely have to check that the function $f$ defined in~\eqref{eq:fDefProofStrictMon} is uniformly continuous. Indeed, $Z:=X-Y$ is uniformly continuous since $X$ and $Y$ are. Thus there exists $\eps>0$ such that $|Z(\omega)-Z(\omega')|<\delta$ whenever $d(\omega,\omega')\leq\eps$. We observe that $d(A_1,A_2)\geq\eps$ and hence
  that the denominator in~\eqref{eq:fDefProofStrictMon} is bounded away from zero. One then checks by direct calculation that $f$ is Lipschitz continuous.

  (iii) We recall that $\mathbb{L}^\infty_\cP(\cFci_t)$ coincides with the set of
  bounded \mbox{$\cP$-quasi} continuous functions (up to modification); cf.~\cite[Theorem~25]{DenisHuPeng.2010}. That is,
  a bounded $\cFci_t$-measurable function $h$ is in $\mathbb{L}^\infty_\cP(\cFci_t)$ if and only if for all $\eps>0$ there exists a closed set $\Lambda\subseteq \Omega$ such that $P(\Lambda)>1-\eps$ for all $P\in\cP$ and such that the restriction $h|_{\Lambda}$ is continuous.

  For $\delta>0$ small enough, we have $P_*(\{X\geq Y+2\delta\})>0$. Then, we can find a closed set $\Lambda\subseteq \Omega$  such that $X$ and $Y$ are continuous on $\Lambda$ and
  \begin{equation}\label{eq:deltaForMonotonicity}
    (1+\delta)\,\E_0(\1_{\Lambda^c})<\delta^2 \E_0(\1_{\{X\geq Y + 2\delta\}\cap\Lambda}).
  \end{equation}
  Define the disjoint closed sets
  \[
    A_1:=\{X\geq Y + 2\delta\}\cap\Lambda,\quad A_2:=\{X\leq Y+\delta\}\cap\Lambda,
  \]
  and let $f$ be the continuous function~\eqref{eq:fDefProofStrictMon}.
  We distinguish two cases. Suppose first that $\delta\E_0(Yf)\leq (1+\delta)\,\E_0(\1_{\Lambda^c})$; then, using~\eqref{eq:deltaForMonotonicity},
  \[
    \E_0(Xf)\geq 2\delta \E_0(\1_{A_1})> (1+\delta)\delta^{-1} \E_0(\1_{\Lambda^c})\geq \E_0(Yf)
  \]
  and we are done. Otherwise, we have $\delta\E_0(Yf)> (1+\delta)\,\E_0(\1_{\Lambda^c})$. Moreover, $\E_0(Xf\1_{\Lambda})\geq (1+\delta)\E_0(Yf\1_{\Lambda})$ can be shown  as in~(i); we simply replace $f$ by $f\1_{\Lambda}$ in that argument. Using the subadditivity of $\E_0$, we deduce that
  \begin{align*}
    \E_0(Xf)+(1+\delta) \, \E_0(Yf\1_{\Lambda^c})
    &\geq \E_0(Xf\1_{\Lambda})+(1+\delta) \, \E_0(Yf\1_{\Lambda^c})\\
    &\geq (1+\delta) \, \E_0(Yf\1_{\Lambda})+(1+\delta) \, \E_0(Yf\1_{\Lambda^c})\\
    &\geq (1+\delta) \, \E_0(Yf)
  \end{align*}
  and hence
  \[
   \E_0(Xf) -\E_0(Yf) \geq  \delta\E_0(Yf) - (1+\delta) \, \E_0(Yf\1_{\Lambda^c}) \geq \delta\E_0(Yf) - (1+\delta) \, \E_0(\1_{\Lambda^c}).
  \]
  The right hand side is strictly positive by assumption.

  (iv) The proof is similar to the one for (iii): we use \cite[Proposition~5.2]{Nutz.10Gexp} instead of \cite[Theorem~25]{DenisHuPeng.2010} to find $\Lambda$, and then the observation made in the proof of~(ii) shows that the resulting function $f$ is uniformly continuous.
\end{proof}

\newcommand{\dummy}[1]{}

\end{document}